\documentclass[reqno]{amsart}

\usepackage{amssymb}
\usepackage{enumerate}

\usepackage[usenames,dvipsnames]{color}
\usepackage[pdftex,bookmarks,colorlinks,breaklinks]{hyperref}  
\definecolor{dullmagenta}{rgb}{0.4,0,0.4}   
\definecolor{darkblue}{rgb}{0,0,0.4}
\hypersetup{linkcolor=red,citecolor=blue,filecolor=dullmagenta,urlcolor=darkblue} 
\hypersetup{linkcolor=black,citecolor=black,filecolor=black,urlcolor=black} 

\newcommand{\eq}[1]{\eqref{#1}}

\newtheorem{theorem}{Theorem}[section]

\newtheorem{lemma}[theorem]{Lemma}
\newtheorem{corollary}[theorem]{Corollary}

\numberwithin{equation}{section}

\DeclareMathOperator{\supp}{supp}
\DeclareMathOperator{\tr}{tr}
\DeclareMathOperator{\Ran}{Ran}
\DeclareMathOperator{\dist}{dist}

\newcommand{\pr}{\prime}

\newcommand\R{\mathbb R}
\newcommand\N{\mathbb N}

\newcommand\Z{\mathbb Z}

\renewcommand\P{\mathbb P}

\renewcommand\H{\mathcal{H}}
\renewcommand\L{\mathrm{L}}
 
\newcommand\D{\mathcal{D}}

\newcommand{\cJ}{\mathcal{J}}
\newcommand{\cD}{\mathcal{D}}

\newcommand{\cG}{\mathcal{G}}
\newcommand{\cE}{\mathcal{E}}

\newcommand{\cF}{\mathcal{F}}

\newcommand{\cV}{\mathcal{V}}

\newcommand{\bPhi}{\boldsymbol{\Phi}}

\newcommand\di{\mathrm{d}}
\newcommand\e{\mathrm{e}}

\newcommand\eps{\varepsilon}
\newcommand{\vrho}{\varrho}
\newcommand\La{\Lambda}
\newcommand{\vphi}{\varphi}

\newcommand\Chi{\raisebox{.2ex}{$\chi$}}

\newcommand{\abs}[1]{\left\lvert #1 \right\rvert}
\newcommand{\norm}[1]{\left\lVert #1 \right\rVert}

\newcommand{\set}[1]{\left\{ #1 \right\}}
\newcommand{\pa}[1]{\left( #1 \right)}

\newcommand{\bra}[1]{\left [ #1 \right ]}

\newcommand{\up}[1]{^{(#1)}}

\newcommand\beq{\begin{equation}}
\newcommand\eeq{\end{equation}}

\newcommand{\qtx}[1]{\quad\text{#1}\quad}

\begin{document}

\title
{Bounds on the  density of states for Schr\" odinger operators}

\author{Jean Bourgain}
\address{Institute for Advanced Study, Princeton, NJ 08540, USA}
 \email{bourgain@ias.edu}

\author{Abel Klein}
\address{University of California, Irvine,
Department of Mathematics,
Irvine, CA 92697-3875,  USA}
 \email{aklein@uci.edu}

\thanks{A.K. was  supported in part by the NSF under grant DMS-1001509.}


\begin{abstract}
We establish bounds on the density of states measure for Schr\" odinger operators. These are deterministic results that do not require the existence of the density of states measure, or, equivalently, of the  integrated density of states. The results are stated in terms of a ``density of states outer-measure" that always exists, and provides an upper bound for the density of states measure when it exists.  We prove log-H\"older continuity for this density of states outer-measure in one, two, and three dimensions for  Schr\" odinger operators, and in any dimension for discrete Schr\" odinger operators.   
\end{abstract}

\maketitle

\tableofcontents

\section{Introduction}

We study the density of states  of the Schr\" odinger operator
\begin{equation} \label{schr}
H = - \Delta + V \qtx{on} \mathrm{L}^2(\mathbb{R}^d), 
\end{equation}
where $\Delta$ is the  Laplacian operator and $V$ is a bounded potential.  The  density of states measure of an interval  gives the ``number of states per unit volume" with energy in the interval; its cumulative distribution function is the integrated density of states.  Finite volume  density of states measures, i.e.,  density of states measures for  restrictions of the Schr\" odinger operator to finite volumes, are always well defined.
The  density of states measure is  given by appropriate limits of  finite volume density of states measures, when such limits exist.  These limits are known to exist for  Schr\"odinger operators where the potential $V$ is in some sense uniform in space (e.g.,   periodic potentials,  ergodic  Schr\"odinger operators), but not for general Schr\"odinger operators.   The  density of states measure and the corresponding   integrated density of states cannot be defined for general Schr\"odinger operators.  For this reason we
introduce the density of states outer-measure, which always exists, and provides an upper bound for  the  density of states measure, when it exists.
We prove upper bounds on the density of states outer-measure of small intervals, establishing  log-H\" older  continuity  in one, two, and three dimensions for  Schr\" odinger operators, and in any dimension for discrete Schr\" odinger operators. 

 We let   
  \beq
\Lambda_{L}(x):= x + \left]-\tfrac L 2,\tfrac L 2\right[^d =\set{y \in \R^d; \; \abs{y-x}_\infty< \tfrac L 2}
\eeq
denote   the  (open)
box of side $L$ centered at $x \in \R^{d}$.  By a box $\Lambda_L$ we will mean a box $\Lambda_{L}(x)$ for some $x \in \R^d$. 
We write     $\norm{\psi}=\norm{\psi}_{2 }$ for $\psi \in  \mathrm{L}^2(\mathbb{R}^d)$ or $\psi \in  \mathrm{L}^2(\La)$.  We set $V_{\infty}=\norm{V}_{\infty}$, the norm of the bounded potential $V$. By $\Chi_{B}$ we denote the characteristic function of the set $B$.  Constants such as $C_{a,b,\ldots}$  will always  be finite and depending only on the parameters or quantities $a,b,\ldots$; they will be independent of other  parameters or quantities in the equation.  Note that $C_{a,b,\ldots}$ may stand for different constants in different sides of the same inequality.

Given a finite box $\Lambda\subset \R^d$, we  let $H_{\Lambda}^\sharp$ and  $\Delta_{\Lambda}^\sharp$ be the restriction of $H$ and $\Delta$ to $\mathrm{L}^2(\Lambda)$ with $\sharp$  boundary condition, where $\sharp = $ $D$ (Dirichlet), $N$ (Neumann), or  $P$ (periodic). We  define finite volume  density of states measures  $\eta_{\La,\sharp}$ on Borel subsets B  of ${\R}^d$  by 
\begin{align}\label{defetasharp}
\eta_{\La,\sharp} (B) &:=\tfrac 1 {\abs{\Lambda}}\tr \set{\Chi_{B}(H^\sharp_\Lambda)}\qtx{for}   \sharp = D, N, P,\\
\eta_{\La,\infty} (B)& :=\tfrac 1 {\abs{\Lambda}} \tr \set{\Chi_{B}(H) \Chi_\Lambda}.
\notag
\end{align}
Note that for for all Borel subsets  $B \subset ]-\infty, E] $ we have
\beq\label{bddsm}
\eta_{\La,\sharp} (B) \le C_{d,V_\infty,E} < \infty   \qtx{for} \sharp =\infty, D,N,P.
\eeq
Moreover, given $f \in C_c(\R)$ and $\delta >0$, there exists $L(d, V_\infty,\delta,f)$ such that for all $L \ge L(d, V_\infty,\delta,f)$ and  $x_0 \in \R^d$ we have
\beq\label{relationf}
\abs{\eta_{\Lambda_L(x_0),\sharp_1}(f) -\eta_{\Lambda_L(x_0),\sharp_2}(f)} \le \delta \qtx{for} \sharp_1,\sharp_2=\infty, D,N,P.
\eeq
(This  can be extracted from \cite[see Theorem~3.6, Theorem~6.2, and their proofs]{DIM}.)  
The finite volume integrated density of states are the corresponding cumulative distribution functions:
\beq
N_{\La,\sharp}(E):= \eta_{\La,\sharp} ( ]-\infty, E] ).
\eeq

For periodic and ergodic Schr\"odinger operators,   density of states measures $\eta_{\sharp}$  can be defined as weak limits of the finite volume density of states measures $\eta_{\La,\sharp}$
for sequences of boxes $\La \to \R^d$ in an appropriate sense. In this case,    the integrated density of states $N_{\sharp}(E):= \eta_{\sharp} ( ]-\infty, E] )$ satisfies $N_{\sharp}(E)=\lim_{\La \to \R^d} N_{\La,\sharp}(E)$ except for a countable set of energies.  Moreover, they all coincide, so we define the  density of states measure $\eta$ and the integrated density of states $N(E)$ by  $\eta(B):=\eta_{\sharp}(B)$ and   $N(E):=N_{\sharp}(E)$ for $\sharp =  \infty, D, N, P$.
(See  \cite{KiM,PF,CL,DIM,Nak}.)

Since infinite volume density of states measures and  integrated density of states cannot be defined for general Schr\"odinger operators,   we
 define density of states outer-measures    on Borel subsets B  of ${\R}^d$    by
 \begin{equation}\label{defetsouter}\begin{split}
 \eta^*_{L,\sharp} (B)& := \sup_{x \in \R^d} \eta_{\La_L(x),\sharp} (B)\\ 
    \eta^*_\sharp (B)& := \limsup_{L\to \infty}  \eta^*_{L,\sharp} (B) \end{split}\; , \qquad \sharp =  \infty, D, N, P .
 \end{equation}
These are always finite on bounded sets in view of \eq{bddsm}. (They are indeed outer-measures, so we call them outer-measures for lack of a better name.)  Moreover,  it follows from \eq{relationf} that
for  all  $E_1,E_2 \in \R$, $E_1 \le E_2$, and $\delta >0$  we have
\beq\label{eta*upper}
\eta^*_{\sharp_1} ([E_1,E_2]) \le \eta^*_{\sharp_2} ([E_1-\delta ,E_2+ \delta]) \qtx{for all} \sharp_1,\sharp_2=\infty, D,N,P.
\eeq

We will say that we have  continuity of  
the  density of states outer-measure $\eta^*_{\sharp}$ if 
\beq \label{eta*upper1}
\lim_{\eps \to 0}\eta^*_{\sharp} ([E-\eps ,E+\eps]) =0 \qtx{for all} E \in \R.
\eeq 
In view of \eq{eta*upper},
 continuity of  
 $\eta^*_{\sharp}$ for some value of $\sharp$ implies  continuity of  
 $\eta^*_{\sharp}$ for all values of $\sharp$, and we have 
\beq   \label{eta*upper2}
 \eta^*_{\infty} ([E_1,E_2])=\eta^*_{D} ([E_1,E_2])=\eta^*_{N} ([E_1,E_2])=\eta^*_{P} ([E_1,E_2])
\eeq
for all  $E_1,E_2 \in \R$, $E_1 \le E_2$.   In this case we set  \beq \label{eta*upper3}
\eta^* ([E_1,E_2]):= \eta^*_{\sharp} ([E_1,E_2]) \qtx{for} \sharp =  \infty, D, N, P,
\eeq
and  say that we have continuity of   the  density of states outer-measure.

We are ready to  state our main result.  Note that if the density of states measure $\eta_{\sharp}$ exists, we always  have
\beq\label{eta<eta*}
\eta_{\sharp}(B) \le  \eta^*_\sharp (B) \qtx{for all Borel sets} B\subset \R^d,
\eeq
and   hence   continuity of  
the  density of states outer-measure implies continuity of the integrated density of states

 \begin{theorem} \label{thmmain}
Let $H$ be a Schr\" odinger operator as in \eq{schr}, where $d=1,2,3$.   Then we have continuity of  
the  density of states outer-measure.  Moreover, given $E_0 \in \R$, for all $E \le E_0 $ and $ \eps \le \tfrac 1 2$ we have 
\beq\label{dosomlog}
\eta^* \pa{[E,E + \eps]} \le \frac {C_{d,V_\infty, E_0 }} {\pa{\log \tfrac 1 \eps}^{\kappa_d} }, \qtx{where}  \kappa_1=1,\;  \kappa_2= \tfrac 1 4, \; \kappa_3= \tfrac 1 8.
\eeq 
\end{theorem}

We also prove a similar result for discrete Schr\" odinger operators, i.e., for 
\begin{equation} \label{schrdiscrete}
H = - \Delta + V \qtx{on} \ell^2(\mathbb{Z}^d), 
\end{equation}
where  $V$ is a bounded potential and  $\Delta$ is the   centered discrete Laplacian,
\beq\label{discLap}
\Delta \psi (x) =\sum_{y \in \Z^d; \; \abs{x-y}=1} \psi(y) \qtx{for} x \in \Z^d.
\eeq
(Our results are still  valid if we take $\Delta$ to be any translation invariant finite range self-adjoint operator on $\ell^2(\mathbb{Z}^d)$.) In $\Z^d$ we define the box of side $L$ centered at $x \in \Z^{d}$ by
\beq
\La= \La_L(x)= \set{y \in \Z^d; \; \abs{y-x}_\infty\le  \tfrac L 2},
\eeq
  and define finite volume operators $H_{\Lambda}^\sharp$ and  $\Delta_{\Lambda}^\sharp$ as the restriction of $H$ and $\Delta$ to $\ell^2(\Lambda)$ with $\sharp$  boundary condition, where $\sharp = $ $D$ (Dirichlet, i.e., simple boundary condition) or  $P$ (periodic).   We  define finite volume  density of states measures  $\eta_{\La,\sharp}$  as in \eq{defetasharp} and  density of states outer-measures  $\eta^*_{L,\sharp},  \eta^*_\sharp $ as in \eq{defetsouter}  for  $\sharp =  \infty, D, P$.  In the discrete case it is easy to see that we also have \eq{eta*upper}, and hence  continuity of  
 $\eta^*_{\sharp}$ for some value of $\sharp$ implies    \eq{eta*upper2}, in which case we  define $\eta^*$ as in \eq{eta*upper3}.

\begin{theorem} \label{thmmaindisc}
Let $H$ be a discrete Schr\" odinger operator as in \eq{schrdiscrete}.  Then for all $d=1,2,\dots$ we have continuity of  
the  density of states outer-measure, and  for all $E \in \R $ and $ \eps \le \tfrac 1 2$ we have 
\beq\label{dosomlogdisc}
\eta^* \pa{[E,E + \eps]} \le \frac {C_{d,V_\infty}} {\log \tfrac 1 \eps}.
\eeq 
\end{theorem}

We are not aware of previous results in the generality  of Theorems~\ref{thmmain} and \ref{thmmaindisc}.  Published results appear to be restricted to cases where we have existence of the integrated density of states. For periodic potentials, continuity of the the integrated density of states is equivalent to the nonexistence of eigenvalues, a nontrivial result proved by  Thomas \cite{Th}.  For ergodic Schr\" odinger operators, continuity of the the integrated density of states is equivalent to the nonexistence of energies that are eigenvalues of infinite multiplicity with probability one (see \cite[Lemma~V.2.1]{CL}). Although  Schr\" odinger operators can have eigenvalues of infinite multiplicity (see \cite{TE}), it is hard to imagine how a fixed energy  can be  an eigenvalue of infinite multiplicity for almost all realizations of an  ergodic Schr\" odinger operator.

Craig and Simon proved log-H\" older continuity (with exponent $1$) of the integrated density of states for  one-dimensional ergodic Schr\" odinger operators \cite{CS1} and for  ergodic discrete  Schr\" odinger operators in any dimension \cite{CS2}.  Delyon and Souillard \cite{DS} provided a simple proof of continuity of the integrated density of states in the discrete case.
 But  continuity of the the integrated density of states for multi-dimensional (continuous) ergodic Schr\" odinger operators, albeit expected, has been hard to prove in full generality.  It is  Problem 14 in   \cite{Siprob}, where it was called (in 2000) a 15 year old open problem.

  For random  Schr\" odinger operators continuity of the integrated density of states follows from a suitable Wegner estimate. The most general result is due to Combes, Hislop and Klopp \cite{CHK} that proved that for the Anderson model, both continuous and discrete, we always have continuity  of the integrated density of states  if the single-site probability distribution has no atoms. (They show that the  integrated density of states has as much regularity as the concentration function of the single-site probability distribution.)
Germinet and Klein \cite{GKloc} proved log-H\" older continuity  of the integrated density of states for the continuous Anderson model with arbitrary single-site probability distribution (e.g., Bernouilli) in the region of localization. (More precisely, in the region of applicability of the multiscale analysis; the    log-H\" older continuity of the integrated density of states is derived  from the conclusions of the multiscale analysis.)

The cases $d=1$ and $d=2,3$ of Theorem~\ref{thmmain} have separate proofs, the proof for $d=1$ being similar to the proof of Theorem~\ref{thmmaindisc}.   Note that  it suffices to establish \eq{dosomlog} and  \eq{dosomlogdisc} with Dirichlet boundary condition ($\sharp =  D$), since we would then have \eq{eta*upper3}.
Thus in the following sections we assume Dirichlet boundary condition  and drop it from the notation.

Theorem~\ref{thmmaindisc} and the $d=1$  case of Theorem~\ref{thmmain} are proved in Section~\ref{secdiscd1}; they are immediate consequences of Theorems~\ref{thmdisc} and  \ref{thmd=1}, respectively. 

Section~\ref{secmultid} is devoted to multi-dimensional Schr\" odinger operators. We start by studying  the local behavior of  approximate solutions of  the stationary Schr\" odinger equation  in Subsection~\ref{subseclocalb}; see Theorem~\ref{lem0psi0}.  Solutions  of  the stationary Schr\" odinger equation admit a local decomposition into a homogeneous harmonic polynomial and a lower order term  \cite{HW,Be}; in  Lemma~\ref{lemBers} we establish  a quantitative  version of this decomposition with   explicit  estimates of the lower order term. This result is extended to approximate solutions in  Lemma~\ref{lem0psi}, implying Theorem~\ref{lem0psi0}.
We then state and prove Theorem~\ref{thmucp}, a   version of 
 Bourgain and Kenig's quantitative unique continuation principle  \cite[Lemma~3.10]{BK},  in  which we make   explicit the dependence on the  parameters relevant to this article. Finally, in 
 Subsection~\ref{subsecd23} we prove Theorem~\ref{thmmaind23}, which implies  the $d=2,3$  cases of Theorem~\ref{thmmain}.

The restriction to $d=1,2,3$ in  Theorem~\ref{thmmain} is  due to the present form of the quantitative unique continuation principle (Theorem~\ref{thmucp}), where there is   a  term $Q^{\frac 4 3}$ in the exponent on the left hand side of  \eq{UCPbound}.  If we had $Q^{\beta}$ in \eq{UCPbound},  we would be able to prove  Theorem~\ref{thmmaind23}, and hence Theorem~\ref{thmmain}, for dimensions
$d < \frac \beta {\beta -1}$.    Since $\beta =\frac 4 3$, we get $d< 4$.  It is reasonable to expect that something like  Theorem~\ref{thmucp} holds with $\beta = 1+$ (there are no counterexamples for real potentials), in which case Theorem~\ref{thmmain} would hold for all $d$, with $\kappa_d=  \frac{\beta- d(\beta -1)} {2\beta}={\frac 1 2 -}$ \,  for $d \ge 2$  \,  in \eq{dosomlog}.

\section{Discrete and  one-dimensional     Schr\" odinger operators}\label{secdiscd1}

To prove Theorem~\ref{thmmaindisc} and the $d=1$  case of Theorem~\ref{thmmain}, we will select a class of approximate eigenfunctions for which we  establish a global upper bound, and use   Lemma~\ref{lemL2Linfty} to pick an  approximate eigenfunction for which we have a lower bound for the global upper bound.  In more detail: 
Given an energy $E$, $0<\eps\le \frac 1 2$, and a box $\La$, we set 
$P=\Chi_{[E,E + \eps]}(H_{\La}) $ and consider the linear space  $\Ran P$. (Note that  $\psi \in \Ran P$  is an approximate eigenfunction for $H_{\La}$ in the sense that $\norm{\pa{H_{\La_L} -E}\psi}\le\eps \norm{\psi}$.)  We select a linear subspace of $\cF$ of $ \Ran P$ for which   the
 $\L^\infty$-norms  are  uniformly bounded  in terms of the $\L^2$-norms (a global upper bound).
 We then use Lemma~\ref{lemL2Linfty} to pick $\psi_0 \in \cF$ for which we have a lower bound for  $\norm{\psi_0}_\infty$. Comparing this lower bound with the global upper bound yields the bound on $\eta_\La \pa{[E,E + \eps]}$.

\subsection{A lower bound for the maximal $\L^\infty$ norm}

\begin{lemma} \label{lemL2Linfty}  Let $\cV$  be a finite dimensional linear subspace of $\, \L^\infty(\Omega,\P)$, where  $(\Omega,\P)$ is a probability space.  Then there exists $\psi \in \cV$ with  $\norm{\psi}_2=1$  such that
\beq
\norm{\psi}_\infty \ge  \sqrt{\dim \cV} .
\eeq
\end{lemma}

This lemma is known to  follow immediately  from the theory of absolutely summing operators,  but can also be proved by  a direct argument.  We present both proofs for completeness.

\begin{proof}[Proof of Lemma~\ref{lemL2Linfty} using absolutely summing operators]

Let   $\cV_p$ denote the linear space $\cV$ viewed as subspace of $\L^p(\Omega,\P)$ and  let $I^{p,q}$ be the identity
map from $\cV_p$ to $\cV_q$, with $\pi_2(I^{p,q})$ being its   $2$-summing norm. (We refer to   \cite{DJA} for the definition and properties of the  $2$-summing norm.)   Then  $\pi_2(I^{2,2})= \sqrt{\dim \cV}$, since    it    is  the same as  the Hilbert-Schmidt norm of $I^{2,2}$.       Factor  $I^{2,2} = I^{2,\infty} I^{\infty,2}$, so  $\pi_2(I^{2,2})\le  \norm{I^{2,\infty}}\pi_2(I^{\infty,2})$ by the  the ideal property  \cite[item 2.4]{DJA}.  Since $\pi_2(I^{\infty,2})\le 1$  \cite[Example~2.9(d)]{DJA}, we have
   $\norm{I^{2,\infty}} \ge \sqrt{\dim \cV}$,  and the lemma follows. 
\end{proof}

\begin{proof}[Proof of Lemma~\ref{lemL2Linfty} (direct proof)] Using the  the Gelfand-Neumark Theorem (e.g.,  \cite[Section 73]{S}) we can assume, without loss of generality,  that $\Omega$ is  a compact Hausdorff space and $\L^\infty(\Omega,\P)=C(\Omega)$.  Thus $\cV$ is a finite dimensional linear subspace  of  $C(\Omega) \subset \L^2(\Omega,\P)$.  Let $N=\dim \cV$, and pick an orthonormal basis  $\set{\phi_j}_{j=1}^N$ for $\cV$.  In particular, 
\beq
\phi (x,y): = \sum_{j=1}^N \overline{\phi_j(x) }\phi_j(y) \in C(\Omega^2),
\eeq
and we  have
\begin{align}
N= \int_\Omega \phi (x,x) \, \P(\di x)=  \int_\Omega\set{ \frac {\phi (x,x)}{\sqrt{\phi (x,x)}} }^2 \P(\di x) \le   \int_\Omega \max_{y \in \Omega} \set{ \frac {\abs{\phi (x,y)}}{\sqrt{\phi (x,x)}} }^2 \P(\di x). 
\end{align}
Since $\P$ is a probability measure, there exists $x_0 \in \Omega$ such that
\beq
 \max_{y \in \Omega}  \frac {\abs{\phi (x_0,y)}}{\sqrt{\phi (x_0,x_0)}} \ge \sqrt{N} .
\eeq
Setting 
\beq
\psi= \frac {\phi (x_0, \cdot)}{\sqrt{\phi (x_0,x_0)}}= \frac {1}{\sqrt{\phi (x_0,x_0)}} \sum_{j=1}^N \overline{\phi_j(x_0) }\phi_j,
\eeq
we have $\psi \in \cV$,  $\norm{\psi}_2=1$, and $\norm{\psi}_\infty \ge  \sqrt{N} $.
\end{proof}

\subsection{Discrete Schr\" odinger operators}
 Theorem~\ref{thmmaindisc} is an immediate consequence of the following theorem.
 
 \begin{theorem}\label{thmdisc} Let $H$ be a discrete Schr\" odinger operator as in \eq{schrdiscrete}.    Then for all $0< \eps\le \frac 12 $ and boxes $\Lambda=\Lambda_L$ with $L \ge L_{d,V_\infty} \, { \log \tfrac 1 \eps}$   we have
 \beq\label{logHolderd1disc}
\eta_\La \pa{[E,E + \eps]} \le \frac {C_{d,V_\infty}} {\log \frac 1 \eps}\qtx{for all} E \in \R.
\eeq 
\end{theorem}

\begin{proof} 
    Let $\La_L=\La_L ( x_0)$ with  $x_0\in \Z^d$, $E \in \R$, and  $\eps \in ]0,\frac 1 2]$.  We  set  $P=\Chi_{[E,E + \eps]}(H_{\La_L}) $, and note that  
     \beq \label{epsineqdisc}
\norm{\pa{H_{\La_L} -E}\psi}_\infty \le   \norm{\pa{H_{\La_L} -E}\psi}\le\eps \norm{\psi}\qtx{for all} \psi \in \Ran P,
\eeq
since we have $ \norm{\psi}_\infty \le  \norm{\psi}$ for  all $ \psi \in\ell^2(\La)$.

We assume 
\beq \label{assumpdisc}
\rho:=\eta_{\La_L} \pa{[E,E + \eps]}= \tfrac 1 { \abs{{\La_L}}} \tr P > 0 ,
\eeq
 since otherwise there is nothing to prove.

 We fix  $R \in 2\N$, $R <L$,  to be selected later, and pick   $\cG\subset {\La_L}$   such that 
 \beq\label{La=uniony}
{\La_L} = \bigcup_{y \in \cG}{\La}_R(y) \qtx{and} {\tfrac { \abs{\La_L}}{\abs{\La_R}}}\le  \# \cG \le  2^d {\tfrac { \abs{\La_L}}{\abs{\La_R}}}.
 \eeq
Note that  $(L-1)^d <  \abs{\La_L}= \pa{2 \left\lfloor \frac L 2\right\rfloor +1}^d\le (L+1)^d$,  and 
$\abs{\La_R}=(R+1)^d$.  We set 
\beq
\partial_2  {\La}_R(y)=  \set{x \in  {\La}_R(y); \; \abs{x-y}_\infty\in \set{\tfrac R 2, \tfrac R 2 -1} },
\eeq
and let
\beq
\partial_R  {\La_L}=\bigcup_{y \in \cG} \partial_2  {\La}_R(y).
\eeq
We have 
 \beq\label{partialR}
 \abs{\partial_2  {\La}_R(y)}\le c_d R^{d-1}, \qtx{so}  \abs{\partial_R  {\La_L} }\le 2^d c_d R^{d-1}{\tfrac { \abs{\La_L}}{\abs{\La_R}}}\le 2^d  c_d \tfrac { \abs{\La_L}} R .
  \eeq
 
 We take
\beq\label{defLL}
L > \tfrac {2^{d+1}  c_d} \rho  + 2, 
\eeq
 since otherwise there is nothing to prove for large $L$, and pick
 \beq\label{discR}
  R\in \left[  \tfrac { 2^{d+1}  c_d} \rho,\tfrac {2^{d+1}  c_d} \rho  + 2\right)\cap 2\N.  
 \eeq

 We now consider the vector space 
 \beq
 \cF=\set{\psi \in \Ran P; \; \psi(x)=0 \qtx{for all} x \in \partial_R  \La_L}.
 \eeq
Since $\cF$ is the vector subspace of $\Ran P$ defined by $ \abs{\partial_R  {\La_L} } $ linear conditions, it follows from \eq{assumpdisc}, \eq{partialR}, and \eq{discR}    that 
 \beq
 \dim \cF\ge  \rho  \abs{\La_L} -   \abs{\partial_R  {\La_L} }  \ge \tfrac 1 2 \rho  \abs{\La_L}\ge  1.
 \eeq
 
  Let $\psi \in \cF$ with $\norm{\psi}=1$.  If $y \in \cG$, it follows from \eq{schrdiscrete}, \eq{discLap}, and  \eq{epsineqdisc} that if  we know that
$\abs{\psi(x)} \le  C$ for all $x$ with   $\abs{x-y}_\infty=k+1, k+2$, then we must have 
$\abs{\psi(x)} \le  C A+\eps$ for  $\abs{x-y}_\infty=k$, where 
$A=2d -1 + \norm{V-E}_\infty$. 
Since $\psi(x)=0$ if $\abs{x-y}_\infty=\tfrac R 2, \tfrac R 2 -1$, we get
\beq
\abs{\psi(x)} \le  \eps\textstyle \sum_{j=0}^{\tfrac R 2 -2- \abs{x-y}_\infty}A^j 
\le \eps\pa{\tfrac R 2 -1} A^{\tfrac R 2 -2}\qtx{for all } x \in \Lambda_{R-4}(y), 
\eeq
so $\abs{\psi(x)} \le\eps \pa{\tfrac R 2 -1} A^{\tfrac R 2 -2}$ for all $x \in  {\La}_R(y)$.
We conclude, using   \eq{La=uniony}, that  
\beq\label{discsup}
\norm{\psi}_\infty \le   \eps \pa{\tfrac R 2 -1} A^{\tfrac R 2 -2} .
\eeq
  
  We now use Lemma~\ref{lemL2Linfty}, obtaining $\psi_0 \in \cF$, $\norm{\psi_0}=1$, such that
\beq\label{psi>disc}
\norm{\psi_0}_\infty \ge \sqrt{\frac {\dim \cF} {\abs{\La_L}}}\ge  \sqrt{ \tfrac 1 2\rho} .
\eeq  
(The volume $\abs{\La_L}$ appears because the measure in  Lemma~\ref{lemL2Linfty} is normalized.)
 Combining   \eq{discsup},  \eq{psi>disc}, and \eq{discR} we get
 \beq
 \sqrt{ \tfrac 1 2\rho} \le  \eps \pa{\tfrac R 2 -1} A^{\tfrac R 2 -2} \le  \eps \tfrac {2^{d}  c_d} \rho A^{\tfrac {2^{d}  c_d} \rho - 1}\le \eps \tfrac {2^{d}  c_d} \rho A^{\tfrac {2^{d}  c_d} \rho } ,
 \eeq  
which implies
\beq
\rho \le \frac  {C_{d, \norm{V-E}_\infty}}  { \log \frac 1 \eps},
\eeq  
which is valid when   \eq{discR} holds, i.e.,      $L \ge {C^\pr_{d, \norm{V-E}_\infty}}  { \log \frac 1 \eps}$.

Since $\sigma\pa{H_{\La_L}}\subset [-2d - V_\infty, 2d +V_\infty]$, we have $\eta_{\La_L} \pa{[E,E + \eps]}=0$ unless $\abs{E} \le 2d +V_\infty + \frac 12$, 
so we get \eq{logHolderd1disc} if  $L \ge L_{d,V_\infty} \,  { \log \tfrac 1 \eps}$.
\end{proof}

\subsection{One-dimensional  Schr\" odinger operators}

The case $d=1$ of Theorem~\ref{thmmain}  is an immediate consequence of the following theorem. Note that one dimensional boxes are  intervals.

\begin{theorem}\label{thmd=1} Let $H$ be a Schr\" odinger operator as in \eq{schr} with  $d=1$. Given $E_0 \in \R$, there exists  $L_{V_{\infty},E_0}$ such that for all $ 0<\eps \le \tfrac 1 2$,  open  intervals $\Lambda=\Lambda_L$ with $L \ge{L_{V_{\infty},E_0}}{\log \frac 1 \eps} $,    and energies $E\le E_0$, we have 
 \beq\label{logHolderd1}
\eta_\La \pa{[E,E + \eps]} \le \frac {C_{V_\infty,E_0 }} {\log \frac 1 \eps}.
\eeq 
\end{theorem}

\begin{proof}
    Let $\La=\La_L =]a_0, a_0 + L[$, $E \in \R$, $\eps \in ]0,\frac 1 2]$.  We  set  $P=\Chi_{[E,E + \eps]}(H_\Lambda) $. Recall that    $\Ran P \Chi_\Lambda \subset \cD (\Delta_{\Lambda})\subset C^1(\La)$ since $d=1$,  and note that  
\beq \label{epsineq1d}
\norm{\pa{H_\La -E}\psi} \le \eps \norm{\psi} \qtx{for all} \psi \in \Ran P.
\eeq

Given  $0<R <L$,  set $a_j= a_0 +j R$ for  $j=1,2,\ldots, \left \lceil \frac L R  \right  \rceil - 1$. We introduce  the  vector space  
\begin{align}
\cF_R : = \set{\psi \in \Ran P ; \;\;  \psi(a_j)=\psi^\pr(a_j)=0  \; \text{for}\;  j=1,2,\ldots,   \left \lceil \frac L R  \right  \rceil - 1}.
\end{align}
Given $\psi \in \cF_R$ and  $j=1,\ldots,  \left \lceil \frac L R  \right  \rceil - 1$,  it follows from Gronwall's inequality (see  \cite{How}), $ \psi(a_j)=\psi^\pr(a_j)=0 $, and \eq{epsineq1d} that for all $x \in ]a_j -R, a_j+R[\cap \La$ we have 
\begin{align}
\abs{\psi(x)} \le \e^{K\abs{x-a_j}} \abs{\int_{a_j}^x  \e^{-K\abs{y-a_j}}\abs{\pa{H_\La -E}\psi(y)} \, \di y}\le  \pa{2 K}^{-\frac 12} \e^{KR} \eps \norm{\psi},
\end{align}
where $K= 1 + \norm{V-E}_\infty$.
Since $\La$ is the union of these intervals, we conclude that
\beq\label{Gronbd}
\norm{\psi }_\infty \le  \pa{2 K}^{-\frac 12} \e^{KR} \eps \norm{\psi} \qtx{for all} \psi \in \cF_R.
\eeq 

We now assume that  
\beq \label{assump1}
\rho:=\eta_{\La_L} \pa{[E,E + \eps]}= \tfrac 1 L \tr P >\tfrac 4 L,
\eeq
 since otherwise there is nothing to prove for large $L$..
Taking   $R=\frac 4 \rho$, it  follows from  \eq{assump1} that
\beq\label{dimF} 
\dim \cF_R \ge \rho L - 2 \pa{ \left \lceil \tfrac L R  \right  \rceil - 1} \ge \rho L - 2\tfrac L R  =  \tfrac 1 2\rho L  >1.
\eeq
Applying Lemma~\ref{lemL2Linfty}, we obtain $\psi_0 \in \cF_R$, $\psi_0 \not = 0$, such that
\beq\label{psi>}
\norm{\psi_0}_\infty \ge \sqrt{\frac {\dim \cF_R} {L}} \norm{\psi_0}\ge  \sqrt{ \tfrac 1 2\rho} \norm{\psi_0}.
\eeq
It follows from \eq{Gronbd} and \eq{psi>} that
\beq
 \sqrt{ \tfrac 1 2\rho}  \le  \pa{2 K}^{-\frac 12} \e^{KR} \eps=  \pa{2 K}^{-\frac 12} \e^{\frac {4K} \rho} \eps .
\eeq

Thus, we get
\beq
\rho \le \frac {8K}  { \log \frac 1 \eps},
\eeq
if $L >  \frac 4 \rho \ge \frac  { \log \frac 1 \eps} {2K} $.

Since $\sigma\pa{H_\La}\subset [ - V_\infty, \infty[$, we have $\eta_\La \pa{[E,E + \eps]}=0$ unless $E \ge -V_\infty - \frac 12$.    Thus,  given  $E_0 \in \R$, there exists  $L_{V_{\infty},E_0}$ such that, for all $ 0<\eps \le \tfrac 1 2$,  open  intervals $\Lambda=\Lambda_L$ with $L \ge {L_{V_{\infty},E_0}}{\log \frac 1 \eps} $,    and energies $E\le E_0$, we have 
 \eq{logHolderd1}. 
 \end{proof}

\section{Multi-dimensional  Schr\" odinger operators}\label{secmultid}

To prove   Theorem~\ref{thmmain} for  $d=2,3$, we will select a class of approximate eigenfunctions for which we can establish uniform local upper bounds, and pick an  approximate eigenfunction for which we have a global  lower bound for the global upper bound.  The  local upper bounds will come from the local behavior of  approximate solutions of  the stationary Schr\" odinger equation (Theorem~\ref{lem0psi0}); the global upper bound will come from the quantitative unique continuation principle (Theorem~\ref{thmucp}).

Given $x\in \R^{d}$ and  $\delta>0$, we set 
$B(x,\delta):= \set{y \in \R^{d}; \, \abs{y-x}<\delta}$.

\subsection{Local behavior of  approximate solutions of  the stationary Schr\" odinger equation}
\label{subseclocalb}

\begin{theorem}\label{lem0psi0}    Let $ \Omega\subset \R^d$ be an open subset,  where $d=2,3,\ldots$, and fix a real valued function $ W\in \L^\infty(\Omega)$,    Let $B(x_0, r_0)\subset \Omega$ for some $x_0\in \R^d$ and $r_0>0$.  Suppose  $\cF $ is  a linear subspace of $H^2(\Omega)$ such that 
\beq\label{Feps0}
\norm{\pa{-\Delta + W}\psi}_{\L^\infty(B(x_0, r_0))} \le  C_\cF  \norm{\psi}_{\L^2(\Omega)}\qtx{for all} \psi \in \cF .
\eeq
Then there exist  constants $\gamma_d>0$ and   $0< r_1=r_1 \pa{d,W_\infty} < r_0$,   where $W_\infty=\norm{W}_{\L^\infty(\Omega)}$, with the property that   for all $N\in \N$ there is a linear subspace $\cF_N$ of  $\cF$, with
\beq
\dim \cF_N \ge \dim \cF  - \gamma_d N^{d-1},
\eeq
 such that for all $\psi \in \cF_N$ we have
\begin{align}\label{phiYN20}
 \abs{\psi(x)}
\le \pa{{C}_{d,{W}_\infty,r_1}^{N^2} \abs{x-x_0}^{N+1} + C_\cF}\norm{\psi}_{\L^2(\Omega)} \qtx{for all} x \in {B}(x_0,  {r_1} ).
  \end{align}

\end{theorem}

We take $d=2,3,\ldots$, and set $\N_0=  \set{0} \cup \N$.  We consider sites $x\in \R^d$, partial derivatives $\partial_j= \frac \partial {\partial x_j}$ for $j=1,2,\ldots,d$, multi-indices  $\alpha \in \N_0^d$, and set
\beq
x^\alpha=\prod_{j=1}^d x_j^{\alpha_j}, \quad D^\alpha= \prod_{j=1}^d \partial_j^{\alpha_j}, \quad \abs{\alpha}=\sum_{j=1}^d  \abs{\alpha_j}, \quad \alpha!= \prod_{j=1}^d \alpha_j!.
\eeq

We let $\H_m\up{d}=\H_m(\R^d)$ denote the vector space of homogenous harmonic polynomials on $\R^d$ of degree $m \in  \N_0$, and recall that \cite[Proposition~5.8 and exercises]{ABR} we have $\dim \H_0\up{d}=1$, $\dim \H_1\up{d}=d$, and, for $m=2,3,\ldots$,
\beq
\dim \H_m\up{d} = \binom{d + m-1}{d-1}- \binom{d + m-3}{d-1}.
\eeq
In particular,  we have
\beq
\dim \H_m\up{2}= 2 \qtx{and} \dim \H_m\up{3}= 2m +1 \qtx{for} m=2,3,\ldots,
\eeq
and  $\dim \H_m\up{d} < \dim \H_{m+1}\up{d}$ for $d>2$. Moreover 
\beq\label{limdim}
\lim_{m\to \infty}  \frac  {\dim \H_m\up{d}} {m^{d -2}}= \frac 2 {(d-2)!}\qtx{for} d \ge 2. 
\eeq
We also  define $ \H_{\le N}\up{d}=\bigoplus_{m=0}^N  \H_m\up{d}$, the vector space of  harmonic polynomials on $\R^d$ of degree $\le N$. It follows from \eq{limdim} that for $d=2,3,\ldots$ there exists a constant $\gamma_d>0$ such that
\beq\label{gammadN}
\dim  \H_{\le N}\up{d} =\sum_{m=0}^N \dim  \H_m\up{d}\le \gamma_d N^{d-1} \qtx{for all }  N \in \N.
\eeq

Let
 \begin{align}
\bPhi(x)= \bPhi_d(x) :=\begin{cases} \pa{d(d-2)\omega_d}^{-1}\abs{x}^{-d +2} & \text{if}\quad d=3,4,\ldots\\
-\tfrac 1{2\pi} \log \abs{x} & \text{if}\quad d=2
\end{cases}.
\end{align}
$\bPhi(x)$ is the fundamental solution to Laplace's equation; $\omega_d$ denotes the volume of the unit ball in $\R^d$.  In particular, 
\beq\label{deltabPhi}
-\Delta \bPhi (x) =\delta(x) \qtx{on} \R^d,
\eeq
and
\beq
\abs{D^\alpha\bPhi(x)}\le C_{d,\abs{\alpha}} \abs{x}^{-d +2- \abs{\alpha}}.
\eeq

Given  $ \Omega=B(x, r)$ for some $x\in \R^d$ and $r>0$ and  $W\in \L^{\infty}(\Omega)$ real valued, we set $W_\infty=\norm{W}_{\L^\infty(\Omega)}$,   and consider the   the stationary Schr\" odinger equation
 \beq\label{Poissoneq}
-\Delta \phi +W \phi =0 \qtx{a.e.\  on} \Omega.
\eeq
We let $\cE_0(\Omega)= \cE_0(\Omega,W)$ denote  the vector subspace formed by solutions $\phi \in H^2(\Omega)$.  
We define linear subspaces  
\beq
\cE_N(\Omega)= \set{\phi \in \cE_0(\Omega); \;\;  \limsup_{x\to x_0} \tfrac {\abs{\phi(x)}}{\abs{x-x_0}^N}< \infty} \qtx{for} N \in \N.
\eeq
  Note  that  $\cE_1(\Omega)=\set{ \phi \in \cE_0(\Omega); \; \phi(x_0)=0}$,  $\cE_N(\Omega)\supset \cE_{N+1}(\Omega)$ for all $N \in \N_0$, and $\cap_{N=0}^\infty \cE_N(\Omega) =\set{0}$ by the unique continuation principle.

A solution of the equation \eq{Poissoneq} admits a local decomposition into a homogeneous harmonic polynomial and a lower order term  \cite{HW,Be}.
The following lemma is a quantitative version of this decomposition; it gives an explicit  estimate of the lower order term.

\begin{lemma}\label{lemBers} Let $ \Omega=B(x_0, 3r_0)$ for some $x_0\in \R^d$ and $r_0>0$,  $d=2,3,\ldots$, and fix  a real valued function $W\in \L^{\infty}(\Omega)$.   For all $N \in \N_0$ there exists a linear map  $Y_N\up{\Omega}\colon \cE_N (\Omega) \to 
 \H_N\up{d}$ such that for all $\phi \in \cE_N(\Omega)$ we have 
 \begin{align}\label{phiYN}
& \abs{\phi(x) - \pa{Y_N\up{\Omega} \phi}(x-x_0)} \\
& \notag   \quad\le r_0^{-\frac d 2}  \pa{r_0^{-1}\widetilde{C}_{d,r_0^2W_{\infty}}}^{N+1} \pa{\tfrac {16} 3 }^{\frac {(N+1)(N+2)}2}\pa{(N+1)!}^{d-2}\abs{x-x_0}^{N+1}\norm{\phi}_{\L^2(\Omega)}
  \end{align} 
 for all $ x \in  \overline{B}(x_0,\tfrac {r_0}2)$.  As a consequence,
 for all   $N \in \N_0$ we have
\beq\label{dimE}
\cE_{N+1}(\Omega)=\ker Y_N\up{\Omega} \qtx{and} \dim \cE_{N+1}(\Omega) \ge  \dim \cE_{N}(\Omega) - \dim  \H_N\up{d}.
\eeq 
In particular, if $ \cJ$ is a vector subspace of $\cE_0(\Omega)$ we have
\beq\label{dimJ}
\dim \cJ \cap \cE_{N+1}(\Omega) \ge \dim \cJ - \gamma_d N^{d-1} \qtx{for all} N \in \N,
\eeq
where $\gamma_d$ is the constant in \eq{gammadN}.
\end{lemma}

\begin{proof}  We  prove the lemma for $\Omega=B(0,3)$; the general case  then follows                by   translating and dilating. We set   $\Omega^\pr=B\pa{0,\frac 32}$, and  note that   $ \phi \in \cE_0$  ($\cE_n= \cE_n(\Omega)$) satisfies   elliptic regularity estimates:
 \beq \label{ellipticreg1}
 \norm{\phi}_{\L^\infty(\Omega^\pr)}  \le  C_{d,W_\infty}  \norm{\phi}_{\L^2(\Omega)}, 
 \eeq
 \beq \label{ellipticreg2}
    \norm{\nabla\phi}_{\L^\infty(B(0,1))}\le C_{d,{W}_\infty} \norm{\phi}_{\L^\infty(\Omega^\pr)} .
 \eeq
 
The estimate \eq{ellipticreg1} follows immediately from   \cite[Theorem~8.17]{GT}.  If we knew  $\phi \in C^2(\Omega)\cap\cE_0$,  the estimate  \eq{ellipticreg2} would follow directly from \cite[Theorem~8.32]{GT}.  To prove  \eq{ellipticreg2} for  arbitrary  $ \phi \in \cE_0$, we fix a mollifier  $\alpha \in C^\infty(\R^d)$ (i.e., $\alpha \ge 0$,  $\int \alpha(x)\,  \di x=1$, $\supp \alpha \subset B(0,1)$),  let $\alpha_n (x)=  n^d \alpha(nx) $ for $n\in \N$, and define $\phi_n =\alpha_n \ast \phi$
on $\R^d$. (We extend $\phi$ to $\R^d$ by $\phi(x)=0$ for $x \notin \Omega$.)  We have $\phi_n\in C^\infty(\R^d)$ and    $\phi_n \to \phi$ in $H^2\pa{\Omega^{\pr}}$.  (See \cite[Chapter~7]{GT}.) 
 Since $ \phi \in \cE_0$, we have
 \begin{align}
-\Delta \phi_n = \alpha_n \ast(-\Delta) \phi = \alpha_n \ast \pa{(-\Delta+W) \phi}-  \alpha_n \ast\pa{W \phi}=-  \alpha_n \ast\pa{W \phi}\qtx{on} \Omega^\pr.
\end{align}
In addition,  setting $\Omega^{\pr\pr}= B(0,\frac 54)$, taking $n\ge 4$, and using Young's inequality for convolutions, we have
\begin{align}\label{convnorms}
\norm{\phi_n}_{\L^\infty(\Omega^{\pr\pr})}&\le \norm{\phi}_{\L^\infty(\Omega^{\pr})},\\
 \norm{(-\Delta+W) \phi_n}_{\L^\infty\pa{\Omega^{\pr\pr}}} &\le \norm{ \alpha_n \ast\pa{W \phi}}_{\L^\infty\pa{\Omega^{\pr\pr}}}+ \norm{ W \phi_n}_{\L^\infty(\Omega^{\pr\pr})}
   \le 2 W_\infty  \norm{\phi}_{\L^\infty(\Omega^{\pr})}. \notag
 \end{align}
Appealing to \cite[Theorem~8.32]{GT}, and using \eq{convnorms}, we get
  \begin{align}  \label{ellipticreg2628}
    \norm{\nabla\phi_n}_{\L^\infty(B(0,1))} &  \le C_{d,{W}_\infty} \pa{\norm{\phi_n}_{\L^\infty\pa{\Omega^{\pr\pr}}}+ \norm{(-\Delta+W) \phi_n}_{\L^\infty\pa{\Omega^{\pr\pr}}}}\\
    & \le C^\pr_{d,{W}_\infty} \norm{\phi}_{\L^\infty(\Omega^\pr)} \qtx{for} n\ge 4. \notag
 \end{align}
 Since  we can find a subsequence $\phi_{n_k}$ such that  $\nabla \phi_{n_k}\to \nabla \phi$\, a.e.\ on $\Omega^{\pr}$, \eq{ellipticreg2} follows from \eq{ellipticreg2628}.

  Given  $ \phi \in \cE_0$   we consider its Newtonian potential given by
\beq
\psi(x) = -  \int_{\Omega^\pr}  W(y) \phi(y) \bPhi(x-y)\, \di y \qtx{for} x \in \R^d.
\eeq
In view of  \eq{ellipticreg1},   we have
\beq\label{Newtpot}
\abs{\psi(x)} \le {W_{\infty}}\norm{\phi}_{\L^{\infty}(\Omega^\pr )} \norm{\bPhi}_{\L^{1}(\Omega)} \le  C_{d,{W}_\infty}{W_{\infty}}  \norm{\phi}_{\L^2(\Omega)}\qtx{for all} x \in  \Omega^\pr .
\eeq

It follows from \eq{deltabPhi} that $\Delta \psi = W  \phi$ weakly in $\Omega^\pr$. 
Thus, letting  $h= \phi- \psi$ we have  $\Delta h=0$ weakly in $\Omega^\pr$, so we conclude that $h$ is a harmonic function in $\Omega^\pr \supset \overline{B}(0,1)$.  In particular (see \cite[Corollary~5.34 and its proof]{ABR}), $h$ is real analytic in  $\Omega^\pr $ and
\beq\label{harmonictaylor}
h(x)=\sum_{m=0}^\infty  p_m(x) \qtx{for all} x \in B(0,1),
\eeq
where  $p_m \in \H_m\up{d}$  for all $m=0,1,\ldots$, and  for  $m=1,2,\ldots$  we have
\beq\label{harmonicderbd}
\abs{p_m(x)}\le C_d\,  m^{d-2} \abs{x}^m \sup_{y \in \partial B(0,1)} \abs{h(y)}\qtx{for all} x \in B(0,1).
\eeq
In addition,  it follows from the mean value property  that for all $y \in \partial B(0,1)$ we have
\begin{align}
\abs{h(y)}   \le \tfrac 1{ \abs{B\pa{y,\frac 12}}} \int_{B\pa{y,\frac 12}} \abs{h(y^\pr)}\, \di y^\pr \le C_{d,{W}_\infty} \norm{\phi}_{\L^2(\Omega)},
\end{align}
using  \eq{ellipticreg1} and  \eq{Newtpot}.
Thus, for all $m=1,2,\ldots$ it follows from \eq{harmonicderbd}  that
\beq\label{harmonicderbd2}
\abs{p_m(x)}\le C_{d,{W}_\infty}m^{d-2} { \norm{\phi}_{\L^2(\Omega)}}   \abs{x}^m \qtx{for all} x \in B(0,1).
\eeq
Setting  $h_{N}= \sum_{m=0}^N  p_m(x)\in  \H_{\le N}\up{d}$, it follows that 
\beq\label{esth}
\abs{h(x) -h_N(x)} \le C_{d,{W}_\infty} { \norm{\phi}_{\L^2(\Omega)}} (N+1)^{d-2} \abs{x}^{N+1} \qtx{for all} x \in \overline{B}\pa{0,\tfrac 12}.
\eeq

For each  $y \in \R^d\setminus\set{0}$ we consider $\bPhi_y(x)=\bPhi(x-y)$, a harmonic function on $\R^d\setminus \set{y}$.  In particular, $\bPhi_y(x)$ is real analytic in $B(0,\abs{y})$, so, 
defining
\beq
  J_m(x ,y)= \sum_{\alpha \in \N_0^d,\, \abs{\alpha}=m} \tfrac 1 {\alpha!}{D^\alpha \bPhi}(y) x^\alpha \qtx{for} x\in  \R^d,
   \eeq 
we have (see \cite{ABR}) 
\beq
\bPhi(x-y)= \bPhi_y(x)= \sum_{m=0}^\infty  J_m(x,y) \qtx{for all} x \in B(0, \abs{y}),
\eeq
the series converging absolutely and uniformly on compact subsets of $B(0, \abs{y})$.
Moreover,   $ J_m(\cdot ,y)\in \H_m\up{d}$, and     for all $ y \in \Z^d$ and   $m=1,2,\ldots$   we have (see \cite[Corollary~5.34 and its proof]{ABR}) that 
\begin{align}\label{Jmbd}
\abs{J_m(x ,y)}\le C_d\, m^{d-2}\pa{\tfrac {4\abs{x}}{3\abs{y}}}^m \sup_{x^\pr \in \partial B\pa{0,\tfrac 3 4\abs{y}}} \abs{\bPhi_y(x^\pr)} \le  C_d\, m^{d-2} \pa{\tfrac {4\abs{x}}{3\abs{y}}}^m\bPhi(\tfrac y 4),
\end{align}
for all $x \in B\pa{0,\tfrac 3 4\abs{y}}$.
Setting $\bPhi_{y,N}(x)= \sum_{m=0}^N J_m(x ,y)\in  \H_{\le N}\up{d}$, it follows that for $ x \in  \overline{B}(0,\tfrac 1 2 \abs{y})$ we have 
\beq\label{phiphiharmonic}
\abs{\bPhi_y(x)- \bPhi_{y,N}(x)}\le C_d (N+1)^{d-2}\pa{\tfrac {4\abs{x}}{3\abs{y}}}^{N+1}\bPhi(\tfrac y 4).
\eeq

We now  proceed by induction.  We define $Y_0\colon \cE_0 \to 
 \H_0\up{d}$ by $Y_0 \phi= \phi(0)$. Given $\phi \in \cE_0$, it follows from the mean value theorem and the elliptic regularity estimates \eq{ellipticreg1} and  \eq{ellipticreg2} that
 \beq
  \abs{\phi(x)-\phi(0)} \le  \sup_{y \in  B(0,1)}   \abs{\nabla\phi(y)}\abs{x} \le C_{d,{W}_\infty}  { \norm{\phi}_{\L^2(\Omega)}} \abs{x} \qtx{for} x \in \overline{B}(0,1).
 \eeq
Thus the lemma holds for $N=0$.

We now let $N \in \N$ and suppose that the lemma is valid for $N-1$. If  $\phi \in \cE_N$, it follows that   $\phi \in \cE_{N-1}$ with $Y_{N-1} \phi=0$, so by the induction hypothesis
\beq\label{indN}
 \abs{\phi(x)} \le C_{N}{ \norm{\phi}_{\L^2(\Omega)}}\abs{x}^N\qtx{for all} x \in   \overline{B}\pa{0,\tfrac 12},
\eeq
where
\beq\label{CN}
C_{N}= \widetilde{C}_{d,{W}_{\infty}}^{N}\pa{\tfrac {16} 3 }^{\frac {N(N+1)}2} \pa{N!}^{d-2}.
\eeq
Using \eq{Jmbd} and \eq{indN}, we  define
\beq
\psi_N(x)= -  \int_{\Omega^\pr} W(y) \phi(y) \bPhi_{y,N}(x)\, \di y\in  \H_{\le N}\up{d}.
\eeq
We  fix $ x \in   \overline{B}\pa{0,\tfrac 12}$  and  estimate
 \begin{align}\label{estpsi1}
\abs{\psi(x)- \psi_{N}(x)} \le 
{W}_{\infty} \int_{\Omega^\pr}\abs{\phi(y)}\abs{ \bPhi_{y,>N}(x)}\, \di y ,
\end{align}
where $ \bPhi_{y,>N}(x)= \bPhi_y(x)-\bPhi_{y,N}(x)$.  Appealing to \eq{phiphiharmonic} and \eq{indN}, we get 
\begin{align}\label{estpsi2}
\int_{\overline{B}(0,\tfrac 12) \setminus {B}(0,2\abs{x})}  \abs{\phi(y)}\abs{ \bPhi_{y,>N}(x)}\, \di y
\le  C_d C_{N}\norm{\phi}_{\L^2(\Omega)} (N+1)^{d-2}\pa{\tfrac 4 3}^{N+1}\abs{x}^{N+1}.
\end{align}
If  $y \notin {B}(0,2\abs{x})$ we have 
$ \abs{y} \ge 2 \abs{x}\ge 1$, and hence, using  \eq{phiphiharmonic}, 
\begin{align}\label{estpsi3}
&\int_{\Omega^\pr \setminus\pa{{B}(0,2\abs{x}) \cup \overline{B}(0,\tfrac 12) }}  \abs{\phi(y)}\abs{ \bPhi_{y,>N}(x)}\, \di y\\ \notag
& \qquad \qquad  \qquad  \le  C_d  (N+1)^{d-2}\pa{\tfrac 4 3}^{N+1}\bPhi(\tfrac 1 4 ) \abs{x}^{N+1}\int_{\Omega^\pr }  \abs{\phi(y)} \, \di y  \\ \notag
&\qquad  \qquad  \qquad  \le  C_d  (N+1)^{d-2}\pa{\tfrac 4 3}^{N+1}\abs{x}^{N+1}{ \norm{\phi}_{\L^2(\Omega)}}.
\end{align}
 Using \eq{Jmbd} and \eq{indN}, we get
 \begin{align}\label{estpsi4}
& \int_{\overline{B}(0,\tfrac 12) \cap  {B}(0,2\abs{x})}  \abs{\phi(y)} \abs{ \bPhi_{y,>N}(x)}\, \di y\\
\notag &\hskip20pt  \le{C_N { \norm{\phi}_{\L^2(\Omega)}}} \int_{\overline{B}(0,\tfrac 12) \cap  {B}(0,2\abs{x})}  \abs{y}^{N}\abs{ \bPhi_{y,>N}(x)}\, \di y  \\
& \notag \hskip20pt \le   
{C_N { \norm{\phi}_{\L^2(\Omega)}}} \int_{\overline{B}(0,\tfrac 12)   \cap B(0,2\abs{x})}  \abs{y}^{N}\abs{\bPhi(x-y)}\, \di y\\
 &\notag \hskip40pt + C_d{C_N { \norm{\phi}_{\L^2(\Omega)}}}\sum_{m=0}^N  m^{d-2}\pa{\tfrac 43 \abs{x}}^m\int_{\overline{B}(0,\tfrac 12)  \cap B(0,2\abs{x})}  \abs{y}^{N-m} \abs{\bPhi(\tfrac y 4)}\, \di y \\
 &\notag \hskip20pt  \le  C_d{C_N { \norm{\phi}_{\L^2(\Omega)}}}\pa{1+  N^{d-2}\pa{\tfrac 43 }^{N+1}}\abs{x}^{N+1},
\end{align}
where we used $\abs{x-y}\le 3 \abs{x}$ for $y \in {B}(0,2\abs{x})$. (Note that we get $\abs{x}^{N+2}$ if $d \ge 3$ and $\abs{x}^{(N+2)- }$ if $d=2$.)
Also  using \eq{Jmbd}, we get
 \begin{align}
 \int_{\Omega^\pr \setminus  \overline{B}(0,\tfrac 12) }  \abs{\phi(y)} \abs{ \bPhi_{y,>N}(x)}\, \di y & \le   
 \int_{\Omega^\pr \setminus  \overline{B}(0,\tfrac 12) }  \abs{\phi(y)} \abs{\bPhi(x-y)}\, \di y\\
 &\notag \hskip-80pt +C_d \sum_{m=0}^N  m^{d-2}\pa{\tfrac 43 \abs{x}}^m \int_{\Omega^\pr  \setminus  \overline{B}(0,\tfrac 12) } \abs{\phi(y)}  \abs{y}^{-m} \abs{\bPhi(\tfrac y 4)}\, \di y \\
 &\notag \le  C_d { \norm{\phi}_{\L^2(\Omega)}}\pa{1+  N^{d-2}\pa{\tfrac 43 }^{N+1}},
\end{align}
where we used $\abs{x}\le \frac 12$.  Since 
$\abs{x} > \frac 1 4$ if $y \in {B}(0,2\abs{x})\setminus  \overline{B}(0,\tfrac 12)$, we obtain
\begin{align}\label{estpsi5}
& \int_{\pa{\Omega^\pr\cap  {B}(0,2\abs{x})} \setminus  \overline{B}(0,\tfrac 12) }  \abs{\phi(y)} \abs{ \bPhi_{y,>N}(x)}\, \di y \\
&  \hskip120pt \le C_d { \norm{\phi}_{\L^2(\Omega)}}\pa{1+  N^{d-2}\pa{\tfrac {16} 3 }^{N+1}}\abs{x}^{N+1}.\notag
 \end{align}
Putting together \eq{estpsi1}, \eq{estpsi2}, \eq{estpsi3}, \eq{estpsi4}, and \eq{estpsi5}, we conclude that  for all $ x \in  \overline{B}(0,\tfrac 12)$ we have  ($C_N\ge 1$)
 \begin{align}\label{estpsiN}
\abs{\psi(x)- \psi_{N}(x)} &\le C_d    C_N {W_{\infty}}(N+1)^{d-2}\pa{\tfrac {16} 3 }^{N+1} \abs{x}^{N+1}   { \norm{\phi}_{\L^2(\Omega)}}.
 \end{align}

We now define $Y_{N} \phi = h_N + \psi_N\in \H_N\up{d} $.  Since $\phi=h + \psi$, for all  $x \in  \overline{B}\pa{0,\tfrac 12}$ it follows from \eq{esth}, \eq{estpsiN}, and \eq{CN},  that 
\begin{align}\label{phiYN1}
& \abs{\phi(x) - \pa{Y_N \phi}(x)} \le \abs{h(x) -h_N(x)} + \abs{\psi(x)- \psi_{N}(x)}\\
 \notag &  \qquad  \le \pa{C_{d,{W}_\infty} +C_d     {W_{\infty}C_N}}(N+1)^{d-2}\pa{\tfrac {16} 3 }^{N+1} \abs{x}^{N+1}   { \norm{\phi}_{\L^2(\Omega)}}
 \\
 \notag &  \qquad  \le \widetilde{C}_{d,{W}_\infty} C_N (N+1)^{d-2}\pa{\tfrac {16} 3 }^{N+1} \abs{x}^{N+1}   { \norm{\phi}_{\L^2(\Omega)}}\\
  \notag &  \qquad  \le \widetilde{C}_{d,{W}_\infty}
  \pa{\widetilde{C}_{d,{W}_{\infty}}^{N}\pa{\tfrac {16} 3 }^{\frac {N(N+1)}2} \pa{N!}^{d-2}} 
  (N+1)^{d-2}\pa{\tfrac {16} 3 }^{N+1} \abs{x}^{N+1}   { \norm{\phi}_{\L^2(\Omega)}}\\
   \notag &  \qquad  \le    \widetilde{C}_{d,{W}_\infty}^{N+1}\pa{\tfrac {16} 3 }^{\frac {(N+1)(N+2)}2} ((N+1)!)^{d-2} \abs{x}^{N+1}   { \norm{\phi}_{\L^2(\Omega)}}, 
   \end{align}
by choosing the constant $ \widetilde{C}_{d,{W}_\infty}$ in \eq{CN}  large enough.  This completes the induction.
 
The lemma is proven, as  \eq{dimE} is an immediate consequence of \eq{phiYN}, and \eq{dimJ} follows from \eq{dimE}  and \eq{gammadN}.
 \end{proof}

Theorem~\ref{lem0psi0} is an immediate consequence   from the following lemma.

\begin{lemma}\label{lem0psi}   Let $ \Omega\subset \R^d$ be an open subset,  where $d=2,3,\ldots$, and fix a real valued function $ W\in \L^\infty(\Omega)$.   Let $B(x_0, r_1)\subset \Omega$ for some $x_0\in \R^d$ and $r_1>0$.  Suppose  $\cF $ is  a linear subspace of $H^2(\Omega)$ such that  
\beq\label{Feps}
\norm{\pa{-\Delta + W}\psi}_{\L^\infty(B(x_0, r_1))} \le  C_\cF  \norm{\psi}_{\L^2(\Omega)}\qtx{for all} \psi \in \cF .
\eeq
Then there exists   $0< r_2=r_2 \pa{d,W_\infty} < r_1$, where $W_\infty=\norm{W}_{\L^\infty(\Omega)}$, with the property that   for all $r \in ]0,r_2]$  there is a linear map  $  Z_r \colon \cF \to \cE_0(B(x_0,r))$ such that
\beq\label{Zbound}
\norm{\psi- Z_r\psi}_{\L^\infty(B(x_0,r))} \le   C_{d,r} C_\cF  \norm{\psi}_{\L^2(\Omega)}\qtx{where} \lim_{r \to 0}  C_{d,r}= 0.
 \eeq
As a consequence, for all $N\in \N$ there is a vector subspace $\cF_N$ of  $\cF$, with
\beq
\dim \cF_N \ge \dim \cF  - \gamma_d N^{d-1},
\eeq
where $\gamma_d$ is the constant in \eq{gammadN}, such that for all $\psi \in \cF_N$ we have
\begin{align}\label{phiYN2}
 \abs{\psi(x)}& \le \pa{\widehat{C}_{d,{W}_\infty,r_1}^{N+1} ((N+1)!)^{d-2}\,  3^{N^2} \abs{x-x_0}^{N+1} + C_\cF}\norm{\psi}_{\L^2(\Omega)}\\
 & \notag \le \pa{{C}_{d,{W}_\infty,r_1}^{N^2} \abs{x-x_0}^{N+1} + C_\cF}\norm{\psi}_{\L^2(\Omega)}
  \end{align}
  for all $x \in \overline{B}(x_0, \tfrac {r_2} 4)$.
\end{lemma}

\begin{proof} It suffices to consider $x_0=0$.  We set $B_r= B(0,r)$.
Given $0<r<r_1$ and  $\psi \in \cF$, we define $Z_r \psi \in \cE_0(B_r)$ as the unique solution $\phi\in {H}^2(B_r)$ to the Dirichlet problem on $B_r$ given by
\beq
\begin{cases}  -\Delta \phi + W \phi =0 \qtx{on}  B_r \\
\phi = \psi \qtx{on} \partial B_r
\end{cases}.
\eeq
This map is well defined in view of \cite[Theorem~8.3]{GT} and  is clearly a linear map.

To prove \eq{Zbound} we will use the Green's function $G_r(x,y)$ for the ball $B_r$. We recall that, abusing the notation by writing $\bPhi(\abs{x})$ instead of $\bPhi(x)$ (see \cite[Section~2.5]{GT}; note that with our definition  $\bPhi(x)= -\Gamma(\abs{x})$),
\beq
G_r(x,y)= \begin{cases}
\bPhi(\abs{x-y}) - \bPhi\pa{\tfrac{\abs{y}} r \abs{x - \tfrac {r^2}{\abs{y}^2} y}} & \text{if}\quad y \not= 0\\
\bPhi(\abs{x}) - \bPhi(r)& \text{if}\quad y = 0
\end{cases}.
\eeq 
Using Green's representation formula  \cite[Eq.~(2.21)]{GT} for $\psi$ and $Z_r\psi$, for all $ x\in B_r$ we have
\begin{align}
\psi(x)  & =   - \int_{\partial B_r} \psi(\zeta) \partial_\nu G_r(x,\zeta) \di S(\zeta) - \int_{B_r}W(y) \psi(y)G_r(x,y) \di y\\
& \notag \hskip100pt  +  \int_{B_r}(-\Delta + W(y)) \psi(y)G_r(x,y) \di y,\\
Z_r \psi(x)  & =   - \int_{\partial B_r} \psi(\zeta) \partial_\nu G_r(x,\zeta) \di S(\zeta) - \int_{B_r}W(y) Z_r\psi(y)G_r(x,y) \di y,
\end{align}
where $\di S$ denotes the surface measure and $\partial_\nu$ is the normal derivative.  Since by an explicit calculation  we have, with $p_2=2$  and $p_d= \frac{d-1}{d-2}$for $d \ge 3$,   that for all $x \in B_r$
\beq
 \norm{G_r(x,\cdot)}_{\L^1(B_{r})}\le C^\pr_d r^{\frac {d(p_d-1)}{p_d}}\norm{G_r(x,\cdot)}_{\L^{p_d}(B_{r})}\le C_d r^{\frac {d(p_d-1)}{p_d}},
\eeq
it follows that 
\begin{align}
&\norm{\psi- Z_r\psi}_{\L^\infty(B_{r})} \\
&\notag \qquad \le C_d r^{\frac {d(p_d-1)}{p_d}}  \pa{ {W}_\infty\norm{\psi- Z_r\psi}_{\L^\infty(B_{r})} + \norm{(-\Delta + W) \psi}_{\L^\infty(B_{r})} }.
\end{align}
Selecting $r_2 \in ]0,r_1[$ such that $C_d r_2^{\frac {d(p_d-1)}{p_d}}(1 + W_\infty) \le \frac 12$, and using \eq{Feps}, we get  \eq{Zbound}.

 Now let  $\cJ = \Ran Z_{r_2}$, a linear subspace  of $ \cE_0(B_{r_2})$; note that 
 \beq  \label{dimJ2}
 \dim \cJ +  \dim \ker  Z_{r_2} =\dim \cF.
 \eeq
We set   $\cJ_N=\cJ  \cap \cE_{N+1}(B_{r_2})$ and $\cF_N =  Z_{r_2}^{-1}(\cJ_N)$. It follows from \eq{dimJ}  and  \eq{dimJ2} that
\beq
\dim \cF_N  =  \dim \ker  Z_{r_2} + \dim \cJ_N \ge \dim \cF - \gamma_d N^{d-1}.
\eeq
If  $\psi \in \cF_N $, we have $ Z_{r_2} \psi \in \cE_{N+1}(B_{r_2})$ and
\beq\label{Zbound1}
\norm{\psi}_{\L^\infty(B_{r_2})} \le \norm{\psi- Z_{r_2} \psi}_{\L^\infty(B_{r_2})} + \norm{Z_{r_2} \psi}_{\L^\infty(B_{r_2})},
\eeq
so \eq{phiYN2} follows from \eq{Zbound} and \eq{phiYN}.
\end{proof}

 \subsection{A quantitative unique continuation principle for approximate solutions of the stationary Schr\" odinger equation}\label{subsecUCP}

We state and prove a  a version of 
 Bourgain and Kenig's quantitative unique continuation principle  \cite[Lemma~3.10]{BK}, in which we make   explicit the dependence on the  parameters relevant to this article.  We give a proof following  \cite[Theorem~A.1]{GKloc}.  

 Given subsets $A$ and $B$ of $\R^{d}$, and  a  function $\vphi$ on   set $B$,  we set $\vphi_{A}:=\vphi \Chi_{A\cap B}$.  In particular, given $x\in   \R^{d}$ and $\delta >0$ we write $\vphi_{x,\delta}: =\vphi_{B(x,\delta )}$.

\begin{theorem}\label{thmucp} Let  ${\Omega}$ be an  open subset  of $\R^d$ and consider    a real measurable function $V$ on ${\Omega}$ with $\norm{V}_{\infty} \le K <\infty$.   Let 
$\psi \in\mathrm{H}^2({\Omega})$ be real valued and  let  ${\zeta} \in \L^2({\Omega})$  be defined by
\beq \label{eq}
-\Delta {\psi} +V{\psi}={\zeta}  \qtx{a.e.\  on} \Omega.
\eeq
 Let   ${\Theta} \subset {\Omega}$  be a bounded measurable set where $\norm{\psi_{\Theta}}_2 >0$. 
Set
\beq \label{defRx0}
{Q}(x,\Theta):= \sup_{y \in \Theta } \abs{y - x} \qtx{for} x \in {\Omega}.
\eeq
Consider $x_0 \in {\Omega}\setminus \overline{\Theta}$ such that
\beq
  \label{xR}
{Q}={Q}(x_0,\Theta)\ge  1 \qtx{and} B(x_0, 6{Q}+ 2)\subset {\Omega}.
\eeq
Then, given
\beq \label{delta}
0<  \delta \le \min\set{   \dist \pa{x_0, {\Theta}},\tfrac 1 {24}},
\eeq
we have
\begin{align} \label{UCPbound}
 \pa{\frac \delta{Q}}^{m \pa{1 + K^{\frac 2 3}}\pa{ {Q}^{\frac 43}  +  \log \frac{\norm{ {\psi}_{{\Omega}}}_{2}} {\norm{ {\psi}_{{\Theta}}}_2}}}\norm{ {\psi}_{{\Theta}}}^2_2  \le   \norm{ {\psi}_{x_0,\delta}}^2_2 + \delta^2 \norm{{\zeta_{{\Omega}}}}_2^2,
\end{align}
where $m>0$ is a constant depending only on $d$.
\end{theorem}

We  we will apply this  theorem with   $\delta \ll 1 \ll{Q}$.

The proof of this theorem is based on the  the  Carleman-type inequality estimate given in  \cite[Lemma~3.15]{BK},   \cite[Theorem~2]{EV},  We state it as in \cite[Lemma~A.5]{GKloc}.

\begin{lemma}\label{lemCarleman}  Given  $\vrho >0$, the function $w_\vrho (x)=\vphi(\frac 1 \vrho \abs{x})$ on $\R^d$, where  $ \vphi(s):= s\,  \e^{-\int_0^s \frac {1 - \e^{-t}} t \, \di t}  $, is a strictly increasing continuous function on $[0,\infty[$, $C^\infty$ on $]0,\infty[$, satisfying 
 \beq\label{wvrho}
\tfrac 1 {C_1 \vrho} \abs{x} \le w_\vrho(x) \le \tfrac 1 { \vrho}\abs{x}   \qtx{for}  x \in B(0,\vrho), \qtx{where} C_1={\vphi(1)}^{-1} \in ]2,3[.
\eeq
Moreover,  there exist positive constants  $C_2$ and $C_3$, depending only on $d$, such that
for all   $\alpha \ge C_2$ and  all real valued functions  $f\in \mathrm{H}^2( B(0,\vrho))$ with $\supp f \subset B(0,\vrho)\setminus \set{0}$ we have
\begin{equation} \label{carlvrho}
\alpha^3 \int_{\R^d}w_\vrho^{-1-2\alpha} f^2  \, \di x \le C_3\,  \vrho^4 \int_{\R^d}  w_\vrho^{2-2\alpha} (\Delta f)^2  \, \di x.
\end{equation}
\end{lemma}

\begin{proof}[Proof of Theorem~\ref{thmucp}] 
Let $x_0 \in {\Omega}\setminus \overline{\Theta}$ satisfy \eq{xR}, where $C_1$ is defined in \eq{wvrho}.  For convenience we may assume $x_0=0$, in which case $  {\Theta}\subset B(0, 2C_1{Q})$,  and  take ${\Omega}= B(0,  \vrho )$, where   $\vrho= 2 C_1{Q}  + 2$.

Let  $\delta$ be as in \eq{delta}, and fix a function $\eta \in C^\infty_{\mathrm{c}}(\R^d)$ given by $\eta(x)= \xi(\abs{x})$, where $\xi $  is an even  $C^\infty$ function  on $\R$, $0\le \xi\le 1$, such that
\begin{gather}
\xi(s) = 1 \; \text{if}  \;  \tfrac {3\delta} 4 \le \abs{s} \le   2C_1{Q}, \quad \xi(s) = 0  \; \text{if}  \;  \abs{s} \le \tfrac \delta 4  \; \text{or}  \;  \abs{s} \ge  2C_1{Q} +1,
\\
\abs{\xi^{(j)}(s)}  \le \pa{\tfrac 4 \delta}^{j} \; \text{if}  \;   \abs{s} \le\tfrac {3\delta} 4, \quad
\abs{\xi^{(j)}(s)}  \le 2^{j} \; \text{if}  \;  \abs{s} \ge  2C_1{Q} , \quad  j=1,2. \notag
\end{gather}
  Note that
$\abs{\nabla \eta (x)} \le  \sqrt{d}\abs{\xi^{\prime}(\abs{x})}$ and $\abs{\Delta \eta (x)} \le  d\abs{\xi^{\prime\prime}(\abs{x})}$.

We will now 
 apply Lemma~\ref{lemCarleman}  to the function $\eta \psi$. In what follows $C_1,C_2,C_3$ are the constants of Lemma~\ref{lemCarleman}, which  depend only on $d$.  By $C_j$, $j=4,5,\ldots$, we will always denote an appropriate nonzero constant depending only on $d$.

Given  $\alpha \ge C_2>1$ (without loss of generality we take  $C_2 >1$),   it follows from \eq{carlvrho} that
\begin{align} \notag
&\frac {\alpha^3}{3 C_3 \vrho^4} \int_{\R^d} w_\vrho^{-1-2\alpha} \eta^2 {\psi}^2  \, \di x \le \tfrac 1 3  \int_{\R^d}  w_\vrho^{2-2\alpha} (\Delta (\eta  {\psi}))^2  \, \di x
\le \int_{\R^d}  w_\vrho^{2-2\alpha}\eta^2 (\Delta  {\psi})^2  \, \di x 
\\ & \hskip40pt  + 4 \int_{\supp \nabla \eta} w_\vrho^{2-2\alpha} \abs{\nabla \eta}^2 \abs{\nabla  {\psi}}^2  \, \di x  +   \int_{\supp \nabla \eta}  w_\vrho^{2-2\alpha} (\Delta \eta)^2  {\psi}^2  \, \di x,
\label{longest} 
\end{align}
where $\supp \nabla \eta \subset \set{\frac \delta 4 \le \abs{x} \le \frac {3\delta} 4}\cup \set{2C_1{Q} \le \abs{x} \le 2C_1{Q} +1}$.

Using \eq{eq}, recalling  $\norm{V}_{\infty} \le K$, and noting that $w_\vrho\le 1 $ on $\supp \eta$,   we get
\beq \label{longest0}
 \int_{\R^d} w_\vrho^{2-2\alpha}\eta^2 (\Delta  {\psi})^2 \, \di x 
 \le  2  K^2 \int_{\R^d} w_\vrho^{-1-2\alpha}\eta^2  {\psi}^2  \, \di x +2 \int_{\R^d}  w_\vrho^{2-2\alpha}\eta^2 {\zeta}^2 \, \di x.
 \eeq
We take
\beq
\alpha_{0}:={\alpha}\rho^{-\frac 4 3} \ge  C_4 \pa{1 + K^{\frac 23}}, \label{alpha0}
\eeq
ensuring   $\alpha > C_2$ and 
\beq \label{alpha01}
 \frac {\alpha^3}{3 C_3 \vrho^4} = \frac {\alpha_0^3}{3 C_3 } \ge  6 {K^2} .
\eeq
As a consequence, using \eq{wvrho} and recalling \eq{defRx0}, we obtain
\beq\label{longest2}
 \int_{\R^d}w_\vrho^{-1-2\alpha} \eta^2  {\psi}^2  \, \di x \ge \pa{\frac { \vrho}{Q}}^{1 + 2 \alpha}\norm{ {\psi}_{{\Theta}}}^2_2  \ge \pa{2C_1}^{1 + 2 \alpha}\norm{ {\psi}_{{\Theta}}}^2_2.
 \eeq
Combining \eq{longest}, \eq{longest0}, \eq{alpha01}, and \eq{longest2},
we  conclude that
\begin{align}
  & \frac {2\alpha_0^3}{9 C_3 }\pa{2C_1}^{1 + 2 \alpha}\norm{ {\psi}_{{\Theta}}}^2_2  \le   4 \int_{\supp \nabla \eta}  w_\vrho^{2-2\alpha} \abs{\nabla \eta}^2 \abs{\nabla  {\psi}}^2  \, \di x  \\ & \hskip20pt \qquad +   \int_{\supp \nabla \eta} w_\vrho^{2-2\alpha} (\Delta \eta)^2  {\psi}^2  \, \di x + 2  \int_{\supp  \eta}  w_\vrho^{2-2\alpha}\eta^2 {\zeta}^2 \, \di x. \notag
\end{align}

We have
 \begin{align}
&  \int_{ \set{{2C_1}{Q} \le \abs{x} \le {2C_1}{Q} +1}}   w_\vrho^{2-2\alpha}\pa{ 4 \abs{\nabla \eta}^2 \abs{\nabla  {\psi}}^2 +  (\Delta \eta)^2  {\psi}^2 } \, \di x \\\notag
&\hskip20pt  \le  4 d^2  \pa{\frac {C_1 \vrho}{{2C_1}{Q}}}^{2\alpha -2} \int_{ \set{{2C_1}{Q} \le \abs{x} \le {2C_1}{Q} +1}}  \pa{ 4  \abs{\nabla  {\psi}}^2 +    {\psi}^2 } \, \di x\\\notag
&\hskip20pt \le  C_5 \pa{\tfrac 5 4 {C_1 }}^{2\alpha -2} \int_{ \set{{2C_1}{Q} -1 \le \abs{x} \le {2C_1}{Q} +2}}  \pa{{\zeta}^2 + (1 +K )  {\psi}^2 } \, \di x\\ \notag
&\hskip20pt \le  C_5 \pa{\tfrac 5 4 {C_1 }}^{2\alpha -2}\pa{\norm{{\zeta_{{\Omega}}}}_2^2 +  (1 +K)\norm{ {\psi}_{{\Omega}}}_{2}^{2}},\notag
\end{align}
where we used an interior estimate (e.g., \cite[Lemma~A.2]{GKduke}).
Similarly,
 \begin{align}
&\int_{\set{\frac \delta 4 \le \abs{x} \le \frac {3\delta} 4}}  w_\vrho^{2-2\alpha}\pa{ 4 \abs{\nabla \eta}^2 \abs{\nabla  {\psi}}^2 +  (\Delta \eta)^2  {\psi}^2 } \, \di x\\ \notag
&  \qquad \qquad  \le  16 d^2 \delta^{-2}\pa{4 \delta^{-1} C_1 \vrho}^{2\alpha -2} \int_{\set{\frac \delta 4 \le \abs{x} \le  \frac {3\delta} 4}}  \pa{ 4  \abs{\nabla  {\psi}}^2 +    {\psi}^2 } \, \di x\\ \notag
& \qquad  \qquad  \le  C_{6}  \delta^{-2} \pa{4  \delta^{-1} C_1 \vrho}^{2\alpha -2} \int_{ \set{ \abs{x} \le  \delta }}    \pa{{\zeta}^2 + (K+ \delta^{-2})  {\psi}^2 } \, \di x \\
& \qquad  \qquad  \le    C_{6}  \delta^{-2}  \pa{16  \delta^{-1} C_1^2{Q}}^{2\alpha -2}\pa{\norm{{\zeta_{{\Omega}}}}_2^2 +  (K+ \delta^{-2})\norm{ {\psi}_{0,\delta}}^2_2}.\notag
\end{align}
In addition, 
\beq
 \int_{\supp  \eta} w_\vrho^{2-2\alpha}\eta^2 {\zeta}^2 \, \di x \le  \pa{4  \delta^{-1}  C_1 \vrho}^{2\alpha -2} \norm{{\zeta_{{\Omega}}}}_2^2\le  \pa{16 \delta^{-1} C_1^2{Q}}^{2\alpha -2}\norm{{\zeta_{{\Omega}}}}_2^2.
\eeq

Thus, if we have 
\begin{align}  \label{keycond}
\alpha_0^{3}\pa{\tfrac 85}^{2 \alpha}\norm{ {\psi}_{{\Theta}}}^2_2 \ge  C_{7} (1 +K)\norm{ {\psi}_{{\Omega}}}_{2}^{2} ,
\end{align}
we obtain
\begin{align}
 & C_5 \pa{\tfrac 5 4 {C_1 }}^{2\alpha -2}(1 +K)\norm{ {\psi}_{{\Omega}}}_{2}^{2} \le \tfrac 1 2  \frac {2\alpha_0^3}{9 C_3 }\pa{2C_1}^{1 + 2 \alpha}\norm{ {\psi}_{{\Theta}}}^2_2,
\end{align}
so we  conclude that
\begin{align} \label{almostbound}
 { \frac {\alpha_0^3}{9 C_3 }}\pa{2C_{1}}^{1 + 2 \alpha}\norm{ {\psi}_{{\Theta}}}^2_2  
  \le C_8  \delta^{-2}\pa{16 \delta^{-1} C_1^2{Q}}^{2\alpha -2}\pa{ (K+ \delta^{-2})\norm{ {\psi}_{0,\delta}}^2_2 + \norm{{\zeta_{{\Omega}}}}_2^2}, 
 \end{align}
 where we used  \eq{delta}. Thus,
\begin{align}
\alpha_0^3{Q}^2 \pa{\pa{8  C_1 {Q}}^{-1}\delta}^{2\alpha }\norm{ {\psi}_{{\Theta}}}^2_2 \le C_9 \pa{(K+ \delta^{-2})\norm{ {\psi}_{0,\delta}}^2_2 +\norm{{\zeta_{{\Omega}}}}_2^2},
\end{align} 
which implies 
\begin{align}\label{almostbound9}
\alpha_0^3 Q^4 \pa{\frac \delta Q}^{4\alpha +4}\norm{ {\psi}_{{\Theta}}}^2_2 \le C_{10} \pa{(1 + K)\norm{ {\psi}_{0,\delta}}^2_2 +\delta^2 \norm{{\zeta_{{\Omega}}}}_2^2},
\end{align}
since $\frac \delta Q \le  \frac 1 {24} < \frac 1 {8  C_1}$ by \eq{delta}.
 
 We now choose $\alpha$. Requiring \eq{alpha0}, to satisfy \eq{keycond}  it suffices to also require
  \beq
 \alpha   \ge C_{11} \pa{1+  \log \frac{\norm{ {\psi}_{{\Omega}}}_{2}} {\norm{ {\psi}_{{\Theta}}}_2}}.
  \eeq
 Thus we can satisfy \eq{alpha0} and \eq{keycond} by taking
\beq
\alpha = C_{12} \pa{1 + K^{\frac 2 3}}\pa{{Q}^{\frac 43}  +  \log \frac{\norm{ {\psi}_{{\Omega}}}_{2}} {\norm{ {\psi}_{{\Theta}}}_2}}. 
 \label{alpha1}
\eeq
Combining with   \eq{almostbound9}, and recalling $Q\ge 1$, we get 
 \begin{align}\label{almostbound967}
& \pa{1 + K^{\frac 2 3}}^3 \pa{\frac \delta Q}^{ C_{13} \pa{1 + K^{\frac 2 3}}\pa{ {Q}^{\frac 43}  +  \log \frac{\norm{ {\psi}_{{\Omega}}}_{2}} {\norm{ {\psi}_{{\Theta}}}_2}}}\norm{ {\psi}_{{\Theta}}}^2_2\\
\notag & \hskip140pt  \le C_{14} \pa{(1 +  K)\norm{ {\psi}_{0,\delta}}^2_2 + \delta^2\norm{{\zeta_{{\Omega}}}}_2^2},
\end{align}
and hence, 
\begin{align}\label{almostbound93454}
 \pa{\frac \delta{Q}}^{m \pa{1 + K^{\frac 2 3}}\pa{ {Q}^{\frac 43}  +  \log \frac{\norm{ {\psi}_{{\Omega}}}_{2}} {\norm{ {\psi}_{{\Theta}}}_2}}}\norm{ {\psi}_{{\Theta}}}^2_2  \le  \norm{ {\psi}_{0,\delta}}^2_2 + \delta^2  \norm{{\zeta_{{\Omega}}}}_2^2,
\end{align}
where $m>0$ is a constant depending only on $d$.
\end{proof}

We will apply Theorem~\ref{thmucp} to approximate eigenfunctions of  Schr\" odinger operators defined on a box $\La$ with Dirichlet boundary condition. In this case  the second condition in \eq{xR}  seems to restrict  the application of  Theorem~\ref{thmucp} to  sites $x_0\in \La$ sufficiently far away from the boundary of $\La$.  But, as noted in \cite[Corollary~A.2]{GKloc}, in this case Theorem~\ref{thmucp} can be extended to sites near the boundary of $\La$
as in the following corollary.

\begin{corollary}\label{corQUCPD}  Consider the Schr\"odinger operator $H_\Lambda:= -\Delta_\Lambda +V $ on $\mathrm{L}^2(\Lambda)$, where  $\Lambda= \Lambda_L(x_0)$ is  the open box of side $L>0$ centered at $x_0 \in \R^d$,  $\Delta_\Lambda$ is the Laplacian with either Dirichlet  or periodic boundary condition on $\Lambda$, and   $V$ a is bounded potential on $\Lambda$ with  $\|V\|_\infty \le K<\infty$. Let    $\psi \in \D(\Delta_\Lambda)$ and fix a   bounded measurable set  ${\Theta} \subset {\La}$ where $\norm{\psi_{\Theta}}_2 >0$.  Set
${Q}(x,\Theta):= \sup_{y \in \Theta } \abs{y - x}$ for $x \in \La$, and  consider  $x_0 \in {\La}\setminus \overline{\Theta}$ such that ${Q}={Q}(x_0,\Theta)\ge  1$.  Then, given
$0<  \delta \le \min\set{   \dist \pa{x_0, {\Theta}},\tfrac 1 {24}}$ such that $B(x_0,\delta)\subset \La$,
we have
\begin{align} \label{UCPboundD}
 \pa{\frac \delta{Q}}^{m \pa{1 + K^{\frac 2 3}}\pa{ {Q}^{\frac 43}  +  \log \frac{\norm{ {\psi}}_{2}} {\norm{ {\psi}_{{\Theta}}}_2}}}\norm{ {\psi}_{{\Theta}}}^2_2  \le   \norm{ {\psi}_{x_0,\delta}}^2_2 + \delta^2 \norm{H_\La \psi}_2^2,
\end{align}
where $m>0$ is a constant depending only on $d$.
\end{corollary}

This corollary is proved exactly as \cite[Corollary~A.2]{GKloc}.

\subsection{Two and three dimensional Schr\" odinger operators}\label{subsecd23}
\begin{theorem}\label{thmmaind23} Let $H$ be a Schr\" odinger operator as in \eq{schr}, where  $d=2,3$. 
Given $E_0 \in \R$, there exists  $L_{d,V_{\infty},E_0}$ such that for all $ 0<\eps \le \tfrac 1 2$,  open boxes $\Lambda=\Lambda_L$ with $L \ge {L_{d,V_{\infty},E_0}}{ \pa{\log \tfrac 1 \eps}^{ \frac {3}8} }$,    and energies $E\le E_0$, we have 
 \beq\label{logHolderd2}
\eta_\La \pa{[E,E + \eps]} \le \frac {C_{d,V_{\infty},E_0}} { \pa{\log \tfrac 1 \eps}^{ \frac {4-d}8} } .
\eeq 
\end{theorem}

\begin{proof}

We fix $\eps \in ]0,\frac 1 2]$, let  $L\ge L_0(\eps)$, where $L_0(\eps) >0$ will be specified later, and take a box    $\La=\La_L $. Since $\sigma (H_\La)\subset [-V_\infty,\infty[$, it suffices to consider   $E_0 \ge -V_\infty -1$ and $E \in [ -V_\infty -1,E_0]$. We   set  $P=\Chi_{[E,E + \eps]}(H_\Lambda) $; note that 
     $\Ran P  \subset \cD (\Delta_{\Lambda})\subset {H}^2(\La)$ and 
\beq \label{epsineq}
\norm{\pa{H_\La -E}\psi} \le \eps \norm{\psi} \qtx{for all} \psi \in \Ran P.
\eeq
Moreover,  for $ \psi \in \Ran P$ we have
\begin{align}
 \norm{\psi}_\infty &=  \norm{\e^{- \pa{H_\Lambda+V_\infty}}\e^{ \pa{H_\Lambda+V_\infty}}\psi}_\infty\\ & \le  \norm{\e^{- \pa{H_\Lambda+V_\infty}}}_{\L^2(\La) \to \L^\infty(\La)}  \norm{\e^{ \pa{H_\Lambda+V_\infty}}\psi}
 \ \le C_{d}\e^{ E_0 + V_\infty +1}\norm{\psi},\notag
\end{align}
where we used that for $t>0$
\beq\label{heatkerneldom}
\norm{\e^{- t\pa{H_\Lambda+V_\infty}}}_{\L^2(\La) \to \L^\infty(\La)}\le \norm{\e^{t\Delta_\Lambda}}_{\L^2(\La) \to \L^\infty(\La)} \le  \norm{\e^{t\Delta}}_{\L^2(\R^d) \to \L^\infty(\R^d)}<\infty.
\eeq
Since $P \pa{H_\La -E}\psi= \pa{H_\La -E}P\psi= \pa{H_\La -E}\psi$ for $ \psi \in \Ran P$, we conclude that
\beq \label{epsineq3}
\norm{\pa{H_\La -E}\psi}_\infty  \le \eps\,  C_{d,V_\infty,E_0} \norm{\psi} \qtx{for all} \psi \in \Ran P.
\eeq

Let
\beq \label{assump23}
\rho:= \eta_\La \pa{[E,E + \eps]}= \tfrac 1 {L^d} \tr P.
\eeq
Recalling the estimate  $ \tr P\le  C_{d,V_\infty,E_0}L^d$  (e.g.,  \cite[Eq.~(A.7)]{GKduke}), where $ C_{d,V_\infty,E_0}\ge 1$, we obtain the uniform upper bound 
\beq\label{trest}
 \rho \le \rho_{\mathrm{ub}}:= C_{d,V_\infty,E_0}\qtx{with} \rho_{\mathrm{ub}}\ge 1 .
\eeq

We assume that
\beq \label{Ldgamma}
L^d >   {2^{3d+1}\gamma_d}\tfrac { \rho_{\mathrm{ub}}}  \rho , 
\eeq
since otherwise there is nothing to prove for $L$ sufficiently large.
Here $\gamma_d$ is the constant  in Theorem~\ref{lem0psi0}; we assume  $2^d \gamma_d \ge 1$ without loss of generality.
We take $R$ such that 
\beq\label{Rreq}
 {2^{d+1}\gamma_d}\tfrac { \rho_{\mathrm{ub}}}  \rho \le  R^d  < \pa{\tfrac L 4}^d;
\eeq
note that
\beq\label{choiceRL}
2 \le \rho   R^d  \qtx{and} 2 \le R^d.
 \eeq
In particular, it follows from  \eq{trest} and \eq{Rreq} that
 \beq\label{approxN}
 N: =  \left\lfloor \pa{\tfrac \rho {2^{d+1}\gamma_d} }^{\frac 1 {d-1}} R^{\frac d {d-1}} \right \rfloor \ge \left\lfloor  \rho_{\mathrm{ub}}^{\frac 1 {d-1}}\right \rfloor \ge 1 .
 \eeq

We  pick  $\cG\subset \La$  such that
 \beq
\overline{\La} = \bigcup_{y \in \cG} \overline{\La}_R(y) \qtx{and}  \# \cG= \pa{\lceil \tfrac L R\rceil} ^d\in \bra{\pa{\tfrac L R}^d,\pa{\tfrac {2L} R}^d}\cap \N.
 \eeq
   Given $y_1 \in \cG$, we apply Theorem~\ref{lem0psi0} with  $\Omega=\La\supset B\pa{y_1, 1}$, $W= V-E$, and $\cF= \Ran P$.  The hypothesis \eq{Feps0} follows from 
 \eq{epsineq3}. We conclude that there exists a vector subspace $\cF_{y_1,N}$ of  $\Ran P$ and $r_0=r_0(d,V_\infty,E_0) \in (0,1) $, such that, using also   \eq{approxN} and \eq{Rreq}, we have
 \begin{align}
\dim \cF_{y_1,N} \ge \rho L^d - \gamma_d N^{d-1} \ge 1  , \label{dimF2399}
\end{align}
and for all $\psi \in\cF_{y_1,N}$   we have 
  \beq\label{Bersestuse55}
 \abs{\psi(y_1+ x)}\le \pa{{C}_{d,{V}_\infty,E_0}^{N^2}  \abs{x}^{N+1} + \eps C_{d,V_\infty,E_0}}\norm{\psi}  \qtx{if} \abs{x} < r_0.
 \eeq
Picking $y_2 \in \cG$, $y_2 \not= y_1$, and applying Theorem~\ref{lem0psi0} with  $\Omega=\La\supset  B\pa{y_2, 1}$, $W= V-E$, and $\cF=\cF_{y_1,N}$, we obtain a vector  vector subspace $\cF_{y_1,y_2,N}$ of $\cF_{y_1,N}$, and hence of $\Ran P$, such that 
  \beq\label{dimF23998}
\dim \cF_{y_1,y_2,N} \ge \dim \cF_{y_1,N} - \gamma_d N^{d-1}  \ge \rho L^d - 2\gamma_d N^{d-1} \ge 1, 
\eeq
and \eq{Bersestuse55} holds for all $\psi \in\cF_{y_1,y_2, N}$ also with $y_2$ substituted for $y_1$.  Repeating this procedure
until we exhaust the sites in $\cG$,
 we conclude that there exists a vector subspace $\cF_R$ of  $\Ran P$ and $r_0=r_0(d,V_\infty,E_0) \in (0,1) $, such that 
  \beq\label{dimF23}
\dim \cF_R \ge \rho L^d - \pa{\tfrac {2L} R}^d\gamma_d N^{d-1}   \ge  \tfrac 1 2\rho L^d \ge  {2^{3d}\gamma_d} \rho_{\mathrm{ub}}  \ge 1 , 
\eeq
where we used the assumption \eq{Ldgamma},  and for all $\psi \in \cF_R$ and  $ y \in \cG$ we have 
  \beq\label{Bersestuse}
 \abs{\psi(y+ x)}\le \pa{{C}_{d,{V}_\infty,E_0}^{N^2}  \abs{x}^{N+1} + \eps C_{d,V_\infty,E_0}}\norm{\psi}  \qtx{if} \abs{x} < r_0.
 \eeq

We let $Q_R$ denote the orthogonal projection onto $\cF_R$.  Since $\tr Q_R= \dim \cF_R$, it follows from \eq{dimF23} that  that we can find a box $\Lambda_{1} =\Lambda_{1}(x_1) \subset\La$ such that
\beq
\tr  {\Chi_{\Lambda_{1}}Q_R \Chi_{\Lambda_{1}}} \ge\pa{ {\left\lceil  L \right\rceil}}^{-d} \tfrac 1 2  \rho L^d \ge 2^{-(d+1)} \rho.
\eeq
But $Q_R=Q_R P= P Q_R$ since $\cF_R \subset \Ran P$, and hence
\beq
 2^{-(d+1)} \rho\le  \tr \Chi_{\Lambda_{1}}Q_R \Chi_{\Lambda_1} =\tr\Chi_{\Lambda_{1}} P  Q_R \Chi_{\Lambda_{1}} 
\le \norm{\Chi_{\Lambda_{1}} P}_1 \norm{ Q_R \Chi_{\Lambda_{1}}}.
\eeq
Recall  that  
\beq
\norm{\Chi_{\Lambda_{1}} P}_1= \norm{P \Chi_{\Lambda_{1}}}_1\le C_{d,V_\infty,E_0}.
\eeq 
 (This is  \cite[Theorem~B.9.2]{SiBull} when $\La=\R^d$.   But by an argument similar to \eq{heatkerneldom}  the crucial estimate  \cite[Eq.~(B11)]{SiBull} holds on  finite boxes $\La$ with constants uniform in $\La$, so a careful reading of the proof of  \cite[Theorem~B.9.2]{SiBull} shows that the result holds  on  finite boxes $\La$ with constants uniform in $\La$.).  We thus conclude that
\beq
 \norm{ Q_R\Chi_{\Lambda_{1}}}  \ge C_{d,V_\infty,E_0}^{\prime} \rho \qtx{with} C_{d,V_\infty,E_0}^{\prime}>0 ,
\eeq
so there exists   $\psi_0 = Q_R \psi_0$ with $\norm{\psi_0}=1$ such that  
\beq \label{psi>23}
\norm{\Chi_{\Lambda_{1}} {\psi_0} } \ge \gamma\rho, \qtx{where} \gamma= \tfrac 1 2  C_{d,V_\infty,E_0}^{\prime}>0 .
\eeq
 (Note that $\gamma \rho <\norm{\psi_0}=1$.)

We   pick $y_0 \in \cG$ such that 
\beq
\tfrac 1 2 \le \tfrac 1 4 R \le \dist\pa{y_0,\Lambda_{1} } 
\le {\tfrac {5  \sqrt{d}} 2 R },
\eeq 
which  can done  by our construction, and apply Corollary~\ref{corQUCPD} with $x_0=y_0$,  $\Theta=\Lambda_{1}$, and potential $V-E$;   note that  (recall \eq{defRx0})
\beq
 \tfrac R 4 + \sqrt{d} \le {Q}={Q}(y_0,\Lambda_{1})\le  {\tfrac {5  \sqrt{d}} 2 R } + \sqrt{d}\le 3 \sqrt{d} R.
\eeq 
 Let  $0<\delta< \delta_0:=\min\set{\frac 1{24}, r_0}$, where $r_0$ is as in \eq{Bersestuse}. It follows from Corollary~\ref{corQUCPD}, using \eq{epsineq},   that
\begin{align} \label{UCPbound23}
 \pa{\tfrac \delta{3\sqrt{d}R}}^{m \pa{1 + K^{\frac 2 3}}\pa{ {R}^{\frac 43}  -  \log \norm{{\psi_0} \Chi_{\Lambda_{1}}}_2  }}\norm{\psi_0 \Chi_{\Lambda_{1}}}_2^2  \le   \norm{ {\psi_0}\Chi_{B(y_0,\delta)}}^2_2 +  \eps^2,
\end{align}
with a constant $m=m_d >0$ and $K= \norm{V-E}_\infty$.  Using \eq{Bersestuse} and \eq{psi>23}, we get
\begin{align} \label{UCPbound234}
 &\pa{\tfrac \delta{3\sqrt{d}R}}^{m \pa{1 + K^{\frac 2 3}}\pa{ {R}^{\frac 43}  -  \log (\gamma \rho )  }}(\gamma \rho)^2  \le  C_d  C_{d,V_{\infty},E_0}^{N^2}  \delta^{2(N+1) +d} +  C_{d,V_\infty,E_0}\,  \eps^2.
\end{align}

  Since  $\rho \ge 2 R^{-d}$ and $\tfrac \delta{3\sqrt{d}R} < \tfrac \delta{3\sqrt{d}} <1$  by \eq{choiceRL},  the inequality \eq{UCPbound234} implies the existence of strictly positive constants  $\widetilde{R}=\widetilde{R}_{d,V_{\infty},E_0}$ and  $M=M_{d,V_{\infty},E_0}  $ such that 
\begin{align} \label{UCPbound99}
\pa{\tfrac \delta{R}}^{M{R}^{\frac 43}   }   \le  C_{d,V_{\infty},E_0} ^{N^2} \delta^{2N } + C_{d,V_\infty,E_0}\,  \eps^2\qtx{for} R \ge \widetilde{R}.
\end{align}
We require 
\beq\label{widehat{R}}
R >  \widehat{R}=  \max\set{\widetilde{R}, \delta_0^{-1}},
\eeq
 and 
choose $\delta$ by (note $C_{d,V_{\infty},E_0} ^{N}\ge 1$)
\beq \label{deltachoice}    
\delta =  \pa{C_{d,V_{\infty},E_0}^N R}^{-1} < \delta_0, \qtx{so} \tfrac \delta{R}=  C_{d,V_{\infty},E_0} ^{N} \delta^2 =  \pa{C_{d,V_{\infty},E_0}^N R^2}^{-1},
\eeq
obtaining
\begin{align} \label{UCPbound991}
\pa{\tfrac \delta{R}}^{M{R}^{\frac 43}   }   \le  \pa{\tfrac \delta{R}}^{N } + C_{d,V_\infty,E_0}\,  \eps^2.
\end{align}

We now take $d=2,3$ and  take $R$ large enough so that 
\beq\label{deltaRN}
 \pa{\tfrac \delta{R}}^{N } \le \tfrac 1 2 \pa{\tfrac \delta{R}}^{M{R}^{\frac 43}}, \qtx{i.e.,}   \pa{C_{d,V_{\infty},E_0}^N R^2}^{N-M{R}^{\frac 43}} \ge 2.
\eeq
To see this,   note that  $\frac 43 < \frac d {d-1}$  for $d=2,3$, so
\beq\label{MRN}
M{R}^{\frac 43} < N=  \left\lfloor \pa{\tfrac \rho {2^{d+1}\gamma_d} }^{\frac 1 {d-1}} R^{\frac d {d-1}} \right \rfloor 
\qtx{if } \rho >  C^\pr_{d,V_{\infty},E_0} R^{\frac {d-4}3},
\eeq
and hence
\beq\label{MRN1}
\pa{C_{d,V_{\infty},E_0}^N R^2}^{N-M{R}^{\frac 43}}\ge 4^{N-M{R}^{\frac 43}}\ge 2 \qtx{if } \rho >  C^{\pr\pr}_{d,V_{\infty},E_0} R^{\frac {d-4}3}.
\eeq

We now choose $R$ by
\beq\label{chooserho}
\rho =c_{d,V_{\infty},E_0} R^{\frac {d-4}3},
\eeq
where the constant $c_{d,V_{\infty},E_0}$ is chosen large enough  to ensure that  all  the conditions \eq{Rreq},   \eq{widehat{R}}, and \eq{MRN1}  (and hence \eq{deltaRN}) are satisfied.  (This can be done using \eq{trest}.)  It then follows  
 from \eq{UCPbound991} and   \eq{deltaRN} that 
 \beq\label{CNR43}
\tfrac 1 2 \pa{\tfrac \delta{R}}^{M{R}^{\frac 43}}\le  C_{d,V_\infty,E_0}\,  \eps^2, \qtx{i.e.,} \pa{C_{d,V_{\infty},E_0}^{N} R^2}^{-M{R}^{\frac 43}} \le 2 C_{d,V_\infty,E_0}\,  \eps^2.
\eeq
Recalling 
  \eq{approxN}, and using   \eq{chooserho} with a sufficiently large constant $c_{d,V_{\infty},E_0}$,   we get from \eq{CNR43} that
\begin{align}
\e^{- M^\pr R^{ \frac 8 3}}= \e^{- M^\pr R^{\frac {d-4}{3(d-1)} + \frac d {d-1} + \frac 4 3}} \le  C_{d,V_\infty,E_0}\,  \eps^2,
\end{align}
where $M^\pr=M^\pr_{d,V_{\infty},E_0}$.  Thus
\beq
\log \tfrac 1 \eps \le C_{d,V_{\infty},E_0} {R}^{\frac 83} = \frac { C_{d,V_{\infty},E_0}^\pr}{ \rho^{\frac 8 {4-d}}},
\eeq
 and hence
\beq
\rho \le  { C_{d,V_{\infty},E_0} }{ \pa{\log \tfrac 1 \eps}^{-\frac {4-d}8}}  ,
\eeq
as long as $L$ is large enough to satisfy \eq{Rreq} with the choice of $R$ in \eq{chooserho}, namely $L \ge L_{d,V_\infty,E_0}  \pa{\log \tfrac 1 \eps}^{\frac 3 8}$.
\end{proof}


\end{document}